\documentclass[SIADSpaper,onefignum,onetabnum]{siamonline250211}

\newcommand{\mytitle}{Learning Koopman Models From Data Under General Noise Conditions}
\usepackage{lipsum}
\usepackage{amsfonts}
\usepackage{graphicx}
\usepackage{epstopdf}
\usepackage{algorithmic}
\ifpdf
  \DeclareGraphicsExtensions{.eps,.pdf,.png,.jpg}
\else
  \DeclareGraphicsExtensions{.eps}
\fi

\usepackage{graphics} % for pdf, bitmapped graphics files
\usepackage{color}
\usepackage{tikz}
\usepackage{subfig}
\usepackage{graphicx}
\usepackage{xcolor}
\usepackage{amsmath} % assumes amsmath package installed
\usepackage{amssymb}  % assumes amsmath package installed
\usepackage{placeins}
\usepackage{commath}
\usepackage{enumerate}
\usepackage[shortlabels]{enumitem}
\usepackage{caption}
\usepackage{multirow}
\usepackage{multicol}
\usepackage{comment}
\usepackage[normalem]{ulem}

\definecolor{blue_def}{HTML}{1f77b4}
\definecolor{orange_def}{HTML}{ff7f0e}
\definecolor{green_def}{HTML}{2ca02c}
\definecolor{red_def}{HTML}{d62728}

\usepackage{enumitem}
\setlist[enumerate]{leftmargin=.5in}
\setlist[itemize]{leftmargin=.5in}

\newsiamremark{remark}{Remark}
\newsiamremark{hypothesis}{Hypothesis}
\crefname{hypothesis}{Hypothesis}{Hypotheses}
\newsiamthm{claim}{Claim}
\newsiamremark{fact}{Fact}
\crefname{fact}{Fact}{Facts}

\headers{\mytitle}{L. C. Iacob, M. Sz\'ecsi, G. I. Beintema, M. Schoukens, and R. T\'oth}

\title{\mytitle\thanks{Submitted to the editors on the 24\textsuperscript{th} of May, 2025. This paper extends \textit{Deep Identification of Nonlinear Systems in Koopman Form}, which has appeared in the Proceedings of the 60th IEEE Conference on Decision and Control, CDC, 2021 
\funding{This work was funded by the European Union (ERC, COMPLETE, 101075836). 
The research was also supported by the European Union within the framework of the National Laboratory for Autonomous Systems (RRF-2.3.1- 21-2022-00002) and by the Air Force Office of Scientific Research under award number FA8655-23-1-7061. Views and opinions expressed are however those of the author(s) only and do not necessarily reflect those of the European Union or the European Research Council Executive Agency. Neither the European Union nor the granting authority can be held responsible for them.}}}

\author{Lucian Cristian Iacob\thanks{ Control System Group, Dept.~of Electrical Engineering, Eindhoven Technical University, The Netherlands 
  (\email{l.c.iacob@tue.nl}, \email{g.i.beintema@tue.nl}, \email{m.schoukens@tue.nl}, \email{r.toth@tue.nl}).}
\and M\'at\'e Sz\'ecsi\thanks{Systems and Control Laboratory, 
HUN-REN Institute for Computer Science and Control, Hungary 
  (\email{szecsi.mate@sztaki.hun-ren.hu}, \email{toth.roland@sztaki.hun-ren.hu}).}
\and Gerben Izaak Beintema\footnotemark[2]
\and \newline Maarten Schoukens\footnotemark[2]
\and Roland T\'oth\footnotemark[2] \textsuperscript{,}\footnotemark[3]}

\usepackage{amsopn}

\ifpdf
\hypersetup{
  pdftitle={\mytitle},
  pdfauthor={L. C. Iacob, M. Sz\'ecsi, G. I. Beintema, M. Schoukens, and R. T\'oth}
}
\fi

\newcommand{\rev}[1]{{\color{black}#1}} % blue
\newcommand{\revtwo}[1]{{\color{black}#1}} % black

\usepackage{euscript}
\newcommand{\Koop}{\EuScript{K}}

\newlist{enumsteps}{enumerate}{1}
\setlist[enumsteps,1]{label=\bfseries Contribution \arabic*: }

\newtheorem{assumption}{Assumption}
\newtheorem{condition}{Condition}
\allowdisplaybreaks

\begin{document}

\maketitle

% REQUIRED
\begin{abstract}
This paper presents a novel identification approach of Koopman models of nonlinear systems with inputs under rather general noise conditions. The method uses deep state-space encoders based on the concept of state reconstructability and an efficient multiple-shooting formulation of the squared loss of the prediction error to estimate the dynamics and the lifted state \revtwo{only} from input-output data. Furthermore, the Koopman model structure includes an innovation noise term that is used to handle process and measurement noise. It is shown that the proposed approach is statistically consistent \revtwo{(estimation error tends to zero when the number of data points goes to infinity)} and computationally efficient due to the multiple-shooting formulation, \revtwo{by which prediction error of the model} can be calculated  \revtwo{on multiple subsections of the data in parallel}. The latter allows for efficient batch optimization of the network parameters and, at the same time, excellent long-term prediction capabilities of the obtained models. The performance of the approach is illustrated by nonlinear benchmark examples \rev{and experimental \revtwo{data from a Crazyflie 2.1} quadcopter}.
\end{abstract}

% REQUIRED
\begin{keywords}
Koopman methods, nonlinear dynamical systems, data-driven modeling, system identification
\end{keywords}

% REQUIRED
\begin{MSCcodes}
37M99, 47B33, 65P99 93B07, 93B15, 93B30
\end{MSCcodes}

\section{Introduction}
Due to the continuously increasing performance expectations for dynamical systems in engineering, nonlinear behavior in many application areas started to become dominant, requiring novel methods that can stabilize and shape the performance of these systems with the ease of conventional approaches that have been developed for linear time-invariant systems.  Hence, recent years have seen a strong push to find global linear embeddings of nonlinear systems to simplify, among others, analysis, prediction and control. One such embedding technique is based on the Koopman framework, where the concept is to lift the nonlinear state space to a (possibly) infinite-dimensional \revtwo{linear} space through the so-called \emph{observable functions}. The dynamics of the original system are preserved and governed by a linear Koopman operator, enabling the representation of the system dynamics via a linear dynamical description \cite{brunton_overview}, \cite{book_koopman} \rev{(for a more in depth-overview of the history of Koopman operator theory and the state-of-the-art see \cite{Mezic_art_21})}. In practice, by choosing a dictionary of a finite number of observables a priori to construct time-shifted data matrices, linear Koopman-based models have been commonly obtained using simple least-squares estimation \cite{Mauroy_sysid}. One such approach, called \emph{dynamic mode decomposition} (DMD) \cite{rowley_dmd}, is based on constructing time-shifted data matrices using directly measured state variables associated with the system. If the dictionary consists of nonlinear functions of the state, this technique is known as \emph{extended DMD} (EDMD) \cite{williams_edmd}. However, among many challenges related to statistical consistency, availability of state-measurements, etc., the main difficulty with these powerful methods lies in choosing a finite number of lifting functions such that, in the lifted state-space, a \emph{linear time-invariant} (LTI) model exists that can capture well the dynamic behavior of the original nonlinear system. While there exist methods for the automatic selection of the observables (see \cite{Sindy}, \cite{yeung_nn}), they still rely on an a priori choice of a dictionary of functions, which many times are difficult to select and even characterize the resulting approximation error by them.
\par To circumvent this, a viable approach has been found in learning the lifting functions from data by the use of machine learning methods such as \rev{Gaussian} processes \cite{Lian_gp}, kernel-based methods \cite{kernel_bounds}, \cite{kernel_edmd}, or various forms of 
\emph{Artificial Neural Networks} (ANNs) \cite{lusch_nn,otto_bilinear,Otto,deepkoco_nn}. Due to their flexibility in describing multiple model structures, applicability to large datasets, many successful applications of these methods have been reported in the literature to obtain accurate and compact Koopman models in practice.  \par 
However, a common drawback of learning-based methods together with the (E)DMD approaches is (i) the assumption of full-state measurement (e.g. \cite{lusch_nn}, \cite{Otto}), which is rarely the case in engineering applications. Some works such as \cite{Ketthong} and \cite{yeung_nn} do address partial state observations, either by only lifting the output \cite{yeung_nn} or by implementing a DMD version that uses time-delayed measurements \cite{Ketthong}. In a different approach, \cite{otto_bilinear} employs a Kalman filter to estimate the lifted state. Nevertheless, a systematic framework for addressing partial state measurements in the context of data-driven Koopman modelling is still lacking at large.
Furthermore, despite the powerful capabilities of these approaches that have been demonstrated in multiple examples, generally (ii) little consideration of measurement or process noise is taken, which can lead to serious bias of the models when applied in real-world applications. Only a few papers present examples where measurement noise is even present in the data (e.g. \cite{Jiang}, \cite{takeishi_noise}) and often only the robustness of the methods is analyzed (e.g. \cite{Sinha}). While there are works that add process noise directly to the lifted representation (e.g. \cite{otto_bilinear}), the way noise enters the Koopman model is merely an assumption. As such, ensuring statistical guarantees of consistency of the estimators remains an open question in the literature.  
A third important issue is that (iii) the estimation of Koopman models for systems with inputs has only recently been investigated, either through a nonlinear lifting \cite{bonnert_nn} or by using state- and input-dependent observables, together with input increments \cite{deepkoco_nn}. 
However, this often leads to models \rev{that have limited applicability}, as it is more difficult to analyze dynamical aspects of the system or to design controllers to regulate the behavior compared to other model classes. \rev{Alternatively}, due to their simple structure, multiple works assume a fully LTI Koopman model (e.g. \cite{Korda}, \cite{Mamakoukas}) or, lately, bilinear (e.g. \cite{Bruder,otto_bilinear,Schulze}). However, it is still largely unexplored how the approximation capability of these model structures in a finite-dimensional setting compares to using more complex \rev{input structures such as control affine or nonlinear in both state and input,} as given in \cite{CDC_Id_Iacob,Aut_Iacob_inputs,Shi_koopman}. \revtwo{Lastly, (iv) while the Koopman framework has been successfully applied to modeling highly complex nonlinear systems, a thorough understanding of the induced approximation error due to the use of finite-dimensional models (projection error) learned on finite data (estimation error) has been lacking. First results for deterministic systems using EDMD have been obtained in \cite{mezic_on_nr_approx} and \cite{zhang_qunatitative}, providing finite data error bounds for ergodic sampling \cite{mezic_on_nr_approx} and i.i.d sampling \cite{zhang_qunatitative}. For finite dictionaries, \cite{zhang_qunatitative} provides projection bounds using a finite element approach. In an extension to systems with inputs, using bilinear models constructed under the assumption of constant input values, \cite{finite_data_err_bounds} provides finite-data error bounds and \cite{schaller_proj_bounds} describes projection error bounds. To circumvent the manual choice of observables, multiple works focus on kernel EDMD \cite{kohne_kedmd_bound}, \cite{kernel_bounds}, \cite{analysis_kernel_edmd}, to provide a characterization of the full approximation error (including both projection and estimation errors), for the autonomous case, and prove convergence in the infinite data limit \cite{kohne_kedmd_bound}. Note that, however, the projection component of the error is usually addressed by assuming the invariance of the Koopman subspace under the chosen kernel function (Wendland \cite{kohne_kedmd_bound}, RBF \cite{kernel_bounds}). There have been extensions to systems with input, either under the assumptions of a bilinear lifted model and specific choice of kernel basis \cite{straesser_bilinear_err_bounds_ieee}, or by approximating a lifted representation with an affine dependency on the input to derive probabilistic error bounds \cite{analysis_kernel_edmd}.
To the best of the authors' knowledge, there are only few results for the noisy setting (\cite{finite_data_err_bounds} describes bounds for stochastic systems and \cite{llamazareselias} proves convergence of autonomous EDMD for noisy measurements). Furthermore, as previously mentioned, full state availability is assumed (not only input/output data) for these methods and the employed data sampling assumptions are  difficult to meet in practice. Moreover, while recent works have focused on kernel basis \cite{kohne_kedmd_bound}, \cite{kernel_bounds}, \cite{straesser_bilinear_err_bounds_ieee}, to name a few, there has been a lack of focus on ANN-based identification of Koopman models even though powerful empirical results have been historically obtained in this setting \cite{lusch_nn}, \cite{otto_bilinear}.}

\par To overcome challenges (i)-\revtwo{(iv)}, we introduce a flexible Koopman model learning method under control inputs, partial measurements, and with statistical guarantees of consistency under process and measurement noise. For this purpose, a deep-learning-based state-space encoder approach is proposed, which is implemented in the   
deepSI toolbox\footnote{deepSI toolbox available at https://github.com/MaartenSchoukens/deepSI} in Python. The main advantages of the approach together with our contributions\footnote{The present paper extends the conference paper \cite{CDC_Id_Iacob} in terms of introducing an innovation noise structure in the Koopman model to handle process and measurement noise, proving the consistency of the estimator, studying various lifted structures for control inputs and providing extensive analysis and testing of the capabilities of the method on benchmarks and real-world data.} are as follows:
\begin{itemize}
\item Analytic derivation of an exact Koopman model \revtwo{structure (both in infinite and finite dimensional form) for nonlinear systems} with control inputs and innovation noise, \revtwo{which} can  \revtwo{represent both} measurement and process  noise \revtwo{in the original system}; 
\item \revtwo{Introducing a} deep-ANN based encoder function \revtwo{which, by learning the reconstruct\-ability map of the Koopman model, can directly provide an} estimate \revtwo{of} the lifted state using \revtwo{only} input-output data (allows %for both full and 
\revtwo{applicability in case of} partial state measurements);
	\item \revtwo{Proposing a} computationally efficient %batch-wise (
    multiple-shooting%optimization 
    -based
    deep-learning \revtwo{approach for} identification %method
    \revtwo{of the introduced} Koopman models;
    \item \revtwo{Showing statistical consistency of the Koopman identification method under general noise conditions, meaning that  the estimation error decreases to zero asymptotically for the identified models with probability 1;}
	\item Comparative study of Koopman model estimation with input structures of different complexities (linear, bilinear, input affine, general);
\end{itemize}
\par The paper is structured as follows. Section \ref{section_behavior} details the general Koopman framework and we discuss the notions of observability and state reconstructability in the Koopman form. The proposed Koopman encoder, the addition of input and the innovation-type model structure are discussed in Section \ref{sec:section_identification} together with the proposed deep-learning-based approach for the estimation of the models. Section \ref{sec:convergence_and_consistency} discusses the convergence and consistency properties of the estimator. In Section \ref{section_experiments}, the approach is tested \revtwo{and analyzed} on \revtwo{the} Wiener-Hammerstien and \revtwo{the} Bouc-Wen benchmarks used in data-driven modeling and on experimental data obtained from a Crazyflie 2.1 quadcopter. %, followed by a discussion of the results. 
The conclusions are presented in Section \ref{sec:section_conclusion}.
%%%%%%%%%%%%%%%%%%%%%%%%%%%%%%%
\section{Preliminaries}\label{section_behavior}
This section introduces the core concept of the Koopman framework and describes the embedding of nonlinear systems in the solution set of linear representations. We show that, while the behavior of the system can be represented using a linear form, a nonlinear constraint still needs to be satisfied on the initial conditions to ensure a one-to-one mapping between the solution sets. Based on this result, we explore observability and constructability concepts in the original and lifted forms, for both autonomous and input-driven systems. 
\subsection{Koopman embedding of nonlinear systems}
First, for the sake of simplicity, consider a discrete-time nonlinear autonomous system:
\begin{equation}\label{eq:nl_aut}
x_{k+1}=f(x_k),
\end{equation}
with $x_k\in\mathbb{R}^{n_\mathrm{x}}$ being the state variable, $f:\mathbb{R}^{n_\mathrm{x}}\rightarrow \mathbb{R}^{n_\mathrm{x}}$ is a bounded nonlinear state transition map and $k\in\mathbb{Z}$ is the discrete time. The initial condition is denoted by $x_0\in\mathbb{X}\subseteq\mathbb{R}^{n_\mathrm{x}}$ and we assume that $\mathbb{X}$ is forward invariant under $\rev{f(\cdot)}$, i.e., $f(\mathbb{X})\subseteq \mathbb{X}$, see \cite{Aut_Iacob_inputs}. The Koopman framework uses observable functions $\phi\in\mathcal{F}$ to lift the system  \eqref{eq:nl_aut} to a higher dimensional space with linear dynamics. These observables $\phi:\mathbb{X}\rightarrow \mathbb{R}$ are scalar functions (generally nonlinear) and are from a Banach function space $\mathcal{F}$. As described in \cite{book_koopman}, the Koopman operator $\Koop:\mathcal{F}\rightarrow\mathcal{F}$ associated with \eqref{eq:nl_aut} is defined through:
\begin{equation}\label{eq:koop_composition}
\Koop\phi=\phi\circ f, \quad \forall \phi \in \mathcal{F},
\end{equation}
where $\circ$ denotes function composition and \eqref{eq:koop_composition} is equal to:
\begin{equation}\label{eq:koop_composition_2}
    \Koop\phi(x_k)=\phi(x_{k+1}).
\end{equation}
Although the Koopman framework typically requires $\mathcal{F}$ to be spanned by an infinite number of basis functions to fully describe the dynamics of \eqref{eq:nl_aut}, for practical use, an $n_\mathrm{f}$-dimensional linear subspace $\mathcal{F}_{n_\mathrm{f}}\subset \mathcal{F}$ is considered, with $\mathcal{F}_{n_\mathrm{f}}=\mathrm{span}\left\lbrace\phi_j\right\rbrace^{n_\mathrm{f}}_{j=1}$. The finite-dimensional approximation of the Koopman operator $\Koop$ can be described using the projection operator $\Pi:\mathcal{F}\rightarrow\mathcal{F}_{n_\mathrm{f}}$, and is given by:
\begin{equation}
\Koop_{n_\mathrm{f}}=\left.\Pi \Koop\right|_{\mathcal{F}_{n_\mathrm{f}}}: \mathcal{F}_{n_\mathrm{f}} \rightarrow \mathcal{F}_{n_\mathrm{f}}.
\end{equation}
In practice, the Koopman matrix representation $A\in\mathbb{R}^{n_\mathrm{f}\times n_\mathrm{f}}$ is commonly used \cite{book_koopman}:
\begin{equation} \label{eq:comp:A}
\Koop_{n_\mathrm{f}} \phi_{j} = \sum_{i=1}^{n_\mathrm{f}} A_{j,i} \phi_{i}.
\end{equation}
\revtwo{Let $\Phi=[ \begin{array}{ccc}
\phi_1 & \cdots &  \phi_{n_\mathrm{f}}
\end{array}]^\top$. Then, the finite-dimensional Koopman operator $\Koop_{n_\mathrm{f}}$ provides an exact embedding of \eqref{eq:nl_aut} (there is no projection error w.r.t. $\Koop$) if the following invariance condition is satisfied:
\begin{equation}
    A\Phi(x)\in\text{span}\{\Phi\}
\end{equation}
or equivalently $\Phi \circ f \in \text{span}\{\Phi\}$.} Note that there exist classes of systems that admit \revtwo{such an} exact finite-dimensional embedding.
% and $\Koop_{n_\mathrm{f}}$ captures exactly the effect of $\Koop$. 
\rev{For example, \cite{Iacob_poly_system} \revtwo{describes the embedding of polynomial systems in a lower triangular form}, building on the well-known example discussed in \cite{Brunton_2016}, \rev{or \cite{otto_bilinear}}.  
\revtwo{Alternatively}, \cite{Iacob_CT_block_embedding} develops an exact \revtwo{finite-dimensional} embedding algorithm for \revtwo{systems that can be represented as a network} interconnection \revtwo{of LTI filter blocks and static polynomial nonlinearities.}  %of block-oriented polynomial systems. 
\revtwo{Furthermore}, \cite{Levine_86,Wang_23,Wong_83} \revtwo{show that specific} classes of nonlinear systems admit an exact \revtwo{finite-dimensional} embedding via immersion.} \revtwo{However, for many systems, using a finite dictionary $\Phi$ can lead to inexact embeddings. The resulting projection error has been treated in works such as \cite{schaller_proj_bounds} or \cite{zhang_qunatitative}.}

We next introduce the lifted state $z_k = \Phi(x_k)$. %, where $\Phi(x_k)=\begin{bmatrix}
%\phi_1(x_k) & \cdots &  \phi_{n_\mathrm{f}}(x_k)
%\end{bmatrix}^\top$. 
The lifted finite dimensional linear representation of \eqref{eq:nl_aut} is then given by:
\begin{equation}\label{eq:koop_aut}
z_{k+1}=Az_k. 
\end{equation}
The main challenge of the Koopman framework is the selection of the observables, including their number, to obtain a suitable approximation in terms of an appropriate norm (or an exact embedding) of the nonlinear system \cite{book_koopman}. Additionally, it is often not clearly stated in the literature that a linear system whose dynamics are governed by the Koopman matrix $A$ is only equivalent in terms of behavior (collections of all solution trajectories) to the original nonlinear system \eqref{eq:nl_aut} if explicit nonlinear constraints are imposed on the initial condition of the lifted state, i.e., equivalent trajectories are only part of a manifold in the extended solution space. We explore this further through a simple example. 
\subsection{Linear representations subject to nonlinear constraints}\label{sec_2b}
To illustrate the concept, we consider the following polynomial system represented by \eqref{eq:nl_aut}, similar to the \rev{well-studied example} described in \cite{Brunton_2016}:
\begin{equation}\label{eq:example_aut}
\begin{bmatrix}
x_{k+1,1} \\x_{k+1,2}
\end{bmatrix}=\begin{bmatrix}
ax_{k,1} \\ bx_{k,2} - cx^2_{k,1}
\end{bmatrix}
\end{equation}
where $a,b,c\in\mathbb{R}$ are constant parameters and $x_{k,i}$ denotes the $i^{\text{th}}$ element of $x_k$. By considering solutions of \eqref{eq:example_aut} only on $[0,\infty )$ with initial condition $x_0\in\mathbb{R}^{2}$, the feasible trajectories are given by:
\begin{equation}
\mathcal{B} =\left\lbrace x:\mathbb{Z}^+_0 \rightarrow\mathbb{R}^2\mid \text{s.t. \eqref{eq:example_aut} is satisfied}\right\rbrace.
\end{equation}
To obtain the Koopman form, the following observables are chosen: $\phi_1(x_k) = x_{k,1}$, $\phi_2(x_k) = x_{k,2}$ and $\phi_3(x_k) = x^2_{k,1}$, which give the equivalent lifted form:  
\begin{equation}\label{eq:lifted_phi_ex}
\Phi(x_{k+1})=\underbrace{\begin{bmatrix}
a & 0 & 0 \\ 0 & b & -c \\ 0 & 0 & a^2
\end{bmatrix}}_A \Phi(x_k). 
\end{equation} 
Based on \eqref{eq:lifted_phi_ex}, consider the system $z_{k+1}=Az_k$ of dimension $n_\mathrm{z}=3$ and with $z_0\in\mathbb{R}^3$. The solution set is described as: 
\begin{equation}\label{eq:lifted_z_set}
\mathcal{B}_\Koop=\left\lbrace z : \mathbb{Z}^+_0\rightarrow\mathbb{R}^3\mid \text{s.t. }z_{k+1}=Az_k\right\rbrace.
\end{equation}
Note that \eqref{eq:lifted_z_set} represents an unrestricted LTI behavior. It is easy to show that $\Phi(\mathcal{B})\subseteq\mathcal{B}_\Koop$, as any $z_k\in\mathcal{B}_{\Koop}$ with $z_0\in\mathbb{R}^3$ for which $z_{0,3}\neq z^2_{0,1}$ will not correspond to a solution of \eqref{eq:example_aut}. %, i.e., $\Phi^{-1}(z_k)=x_k\notin\mathcal{B}$. 
By introducing the constraint $\Psi:\mathbb{R}^3\rightarrow\mathbb{R},\Psi(z_k)=z^2_{k,1}-z_{k,3}$, the solution set \eqref{eq:lifted_z_set} \revtwo{becomes} %with constraint $\Psi$ 
%is:
\begin{equation}
\hat{\mathcal{B}}_{\Koop}=\left\lbrace z : \mathbb{Z}^+_0\rightarrow\mathbb{R}^3\mid\text{s.t. }z_{k+1}=Az_k,\;\Psi(z_0)=0\right\rbrace .
\end{equation}
Then, it is possible to show that $\Phi(\mathcal{B})=\hat{\mathcal{B}}_\Koop$ holds. To illustrate this, Fig.~\ref{fig:compliant_embedding} shows the simulated trajectories of system  \eqref{eq:lifted_z_set} with $a=0.99$, $b=0.9$ and $c=0.9$ and the constraint $\Psi$, which we call the \emph{compliant surface}.  As can be seen in Fig.~\ref{fig:compliant_embedding}, only solutions (in green) starting on the compliant surface remain on the compliant surface and correspond to solutions (in black) of the original nonlinear system \eqref{eq:example_aut}. 
\rev{This example} highlights the need for additional constraints on the Koopman form, or, as we now call it, the embedding of \eqref{eq:nl_aut}, to guarantee a bijective relationship between the solution sets.  

\par 
Note that, when $x_0$ is known and the observable set $\Phi$ is given, this nonlinear condition on the lifted states is alternatively defined by $z_0=\Phi(x_0)$. \rev{This approach has been extensively used in the Koopman literature in both theoretical and application oriented works, see e.g. \cite{bonnert_nn,Cisneros20,Mauroy_sysid,surana_obs} and, for a discussion on the connection between the nonlinear and lifted manifolds for systems of the type \eqref{eq:example_aut}, one can consult \cite{Brunton_2016}.} However, in an identification setting where information on $x_0$ is not available or only partially available \rev{(in contrast to explicitly lifting the state via $\Phi(x_0)$ as done in \cite{Brunton_2016})}, to construct a lifted model with the constrained solution set, one needs to include the $\Psi(z_0)=0$ condition. Next, we explore observability and reconstructability of $z$ in view of \rev{this discussion.}

\begin{figure}[ht!]
  \centering
\includegraphics[scale=0.7]{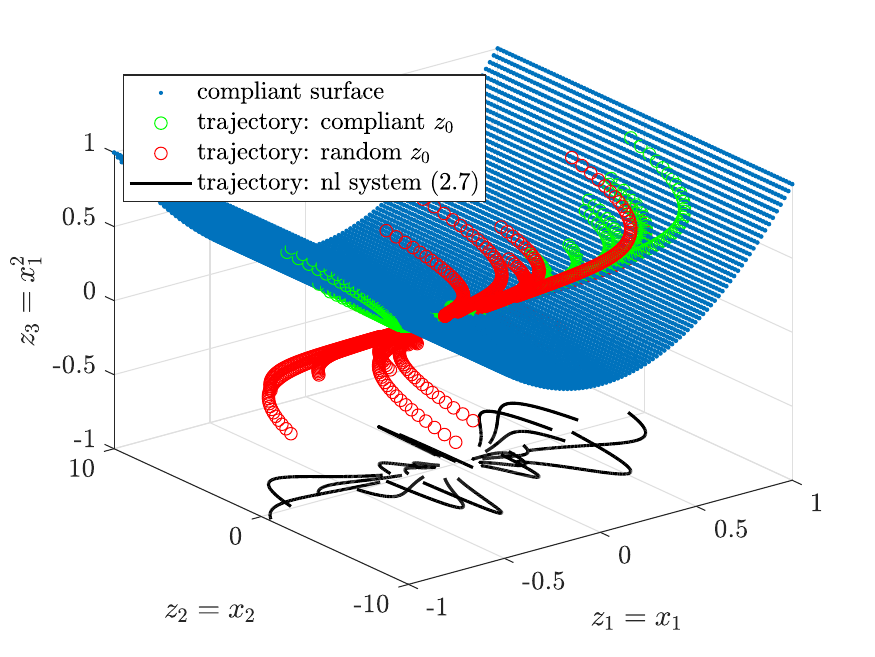} 
  \caption{Compliant surface corresponding to $\Psi$ (in blue), compliant trajectories of the lifted system (in green), non-compliant trajectories (in red) of the lifted system and trajectories of the original nonlinear system (in black).}\vspace{-1cm} \label{fig:compliant_embedding}
\end{figure}
\subsection{Observability and reconstructability}\label{sec:obs_and_reconstr_aut}
Consider the system defined by \eqref{eq:nl_aut} with output 
\begin{equation}\label{eq:nl_out}
    y_k = h(x_k),
\end{equation}
where $h:\mathbb{R}^{n_\mathrm{x}}\rightarrow\mathbb{R}^{n_\mathrm{y}}$ is a nonlinear output map. Given $x_0\in\mathbb{R}^{n_\mathrm{x}}$, the \rev{output map} for the state $x$ associated with the nonlinear dynamics represented by \eqref{eq:nl_aut} and \eqref{eq:nl_out} is: 
\begin{equation}\label{eq:first_obs_map_nl_aut}
\mathcal{O}_{\mathrm{x},n}(x_0)=\begin{bmatrix}
h(x_0) \\ h\circ f(x_0) \\ \vdots \\ h\circ_{n-1}f(x_0)
\end{bmatrix}=\begin{bmatrix}
y_0 \\ y_{1} \\ \vdots \\ y_{n-1}
\end{bmatrix}
\end{equation}
where $h\circ_2 f(x_0)=h\circ f \circ f(x_0)$ and $h\circ_n f(x_0)=h\circ f \circ_{n-1} f(x_0)$ for $n>2$. Let $y^{n-1}_0=[\begin{array}{ccc}
y^\top_0 & \cdots & y^\top_{n-1}
\end{array}]^\top$. As described in \cite{nl_obs_1982}, the representation satisfies the so-called
observability rank condition at $x_0$ if, for $n= n_\mathrm{x}$, the rank of the Jacobian of $\mathcal{O}_{\mathrm{x},n}$ at $x_0$ is $n_\mathrm{x}$. 
%of the analytic map $\mathcal{O}_{\mathrm{x},n}:\mathbb{R}^{n}\rightarrow\mathbb{R}^{n n_\mathrm{y}}$ is \tb{$n_\mathrm{x}$}.
If this condition is met, the representation is strongly locally observable at $\mathbb{X}_0$, where $\mathbb{X}_0$ is a neighborhood of $x_0$, and  
$\mathcal{O}_{\mathrm{x},n}$ is a diffeomorphism, i.e., it is
invertible, on $\mathbb{X}_0$ \cite{Isidori}. We denote its inverse as $\Lambda_{\mathrm{x},n}:\mathbb{R}^{n n_\mathrm{y}}\rightarrow \mathbb{X}_0\subseteq \mathbb{R}^{n_\mathrm{x}}$, such that
 \begin{equation} \label{inv:obs}
 	x_0=\Lambda_{\mathrm{x},n} (\mathcal{O}_{\mathrm{x},n}(x_0))=\Lambda_{\mathrm{x},n}(y^{n-1}_0),
\end{equation}  
for all $x_0\in\mathbb{X}_0$. Note that just like in the LTI case, if this property is satisfied for $n= n_\mathrm{x}$, then (i) the rank condition can be satisfied for $ n_\mathrm{x} \geq n\geq 1$ and the minimal $n$ for which it holds is called the lag $n_\ast$ of the system at $x_0$; (ii) the rank of the Jacobian does not change for $n\geq n_\mathrm{x}$; (iii) hence, the existence of \eqref{inv:obs} is ensured for any $n\geq n_\ast$.

\rev{Throughout the paper, we call $\Lambda_{\mathrm{x},n}$ the observability map, which can be used to} compute the initial condition $x_0$ from future values of the output $y$. Conversely, the reconstructability concept is used to calculate the initial condition from past values of the output. For $n\geq 1$, it holds that
\begin{equation}\label{eq:nl_interm_reconstr_aut}
	x_0 = \circ_{n-1}f(x_{-n+1}),
\end{equation}
where $\circ_0 f(x_0) =x_0$. 
If, for the given $n$, $\Lambda_{\mathrm{x},n}$ exists, then, using \eqref{eq:nl_interm_reconstr_aut}, we have:
\begin{equation}
	x_0 = \circ_{n-1}f(\Lambda_{\mathrm{x},n}(y^0_{-n+1}))\rev{=:}\Pi_{\mathrm{x},n}(y^0_{-n+1})
\end{equation}
for all $x_0\in\mathbb{X}_0$. The function $\Pi_{\mathrm{x},n}:\mathbb{R}^{n n_\mathrm{y}}\rightarrow\mathbb{X}_0\subseteq \mathbb{R}^{n_\mathrm{x}}$ is called the reconstructability map. \par In the Koopman form, assuming that the output function is in the span of the lifted states (or simply included in the lifting set), i.e., $y_k=Cz_k$, with $C\in\mathbb{R}^{n_\mathrm{y}\times n_\mathrm{z}}$, we construct the following map: 
\begin{equation}\label{eq:nl_alg_constr}
\begin{bmatrix}
y_0 \\ y_{1} \\ \vdots \\ y_{n-1} \\ 0
\end{bmatrix} =\begin{bmatrix}\begin{pmatrix}
C \\ CA \\ \vdots \\ CA^{n-1} \end{pmatrix}z_0 \\ \Psi(z_0) 
\end{bmatrix}:=\begin{bmatrix} \mathcal{O}^{\text{lin}}_{\mathrm{z},n} z_0\\ \Psi(z_0)\end{bmatrix}:=\mathcal{O}_{\mathrm{z},n}(z_0).
\end{equation}
In an LTI sense, the lifted system representation would be observable if there is an $n\geq 1$ such that %for $n\geq n_\mathrm{z}$, 
$\text{rank}(\mathcal{O}^{\text{lin}}_{\mathrm{z},n})=n_\mathrm{z}$. However, as observed in Section \ref{sec_2b}, it is also necessary to consider the nonlinear constraints $\Psi:\mathbb{R}^{n_\mathrm{z}}\rightarrow\mathbb{R}^{n_\mathrm{c}}$ to ensure compliance of the initial condition $z_0$. Hence, if there is an $n\geq 1$
such that 
%Then, if the Jacobian $\nabla_{z_0} \Psi$ does not introduce row dependencies that cancel directions in $\mathcal{O}^{\text{lin}}_{\mathrm{z},n}$ and, for $n\geq n_\mathrm{z}$, 
%$\text{rank}(\mathcal{O}^{\text{lin}}_{\mathrm{z},n})=n_\mathrm{z}$, then 
the Jacobian of $\mathcal{O}_{\mathrm{z},n}(z_0)$  has full rank $n_\mathrm{z}$, which implies that the mapping $\mathcal{O}_{\mathrm{z},n}$ is locally invertible on a neighborhood $\mathbb{Z}_0$ of $z_0$, then $z_0$ is uniquely determined from $y^{n-1}_0$ and the constraint $\Psi(z_0)$. Then, there exists a $\Lambda_{\mathrm{z},n}:\mathbb{R}^{n n_\mathrm{y}}\rightarrow \mathbb{Z}_0\subseteq \mathbb{R}^{n_\mathrm{z}}$, such that
\begin{equation}\label{eq:koopman_obs_map_aut}
	z_0=\Lambda_{\mathrm{z},n}(y^{n-1}_0),
\end{equation}
for all $z_0\in\mathbb{Z}_0$.
We call $\Lambda_{\mathrm{z},n}$ the observability map for autonomous Koopman models. To utilize only past data for determining $z_0$, we can also formulate \eqref{eq:koopman_obs_map_aut} in a reconstructability form. Let:
\begin{equation}\label{eq:koop_interm_reconstr_aut}
	z_0 = A^{n-1}z_{-n+1}.
\end{equation}
Using \eqref{eq:koopman_obs_map_aut} in \eqref{eq:koop_interm_reconstr_aut}:
\begin{equation}
	z_0=A^{n-1}\Lambda_{\mathrm{z},n}(y^0_{-n+1}):=\Pi_{\mathrm{z},n}(y^0_{-n+1})
\end{equation}
where $\Pi_{\mathrm{z},n}:\mathbb{R}^{nn_{\mathrm{y}}}\rightarrow\mathbb{R}^{n_\mathrm{z}}$ is the Koopman reconstructability map for autonomous systems. Note that the compliance constraint $\Psi$
is part of $\Lambda_{\mathrm{z},n}$. This gives a different point of view on reconstructability than in the work \cite{surana_obs}, where the observability notions are discussed based on an explicit definition of the lifting map, i.e., a given selection of the observables $\Phi$. Note that, for a nonlinear system representation with $n_\mathrm{x}$ states and an explicit dictionary $\Phi$, the construction $z_0=\Phi(x_0)=\Phi(\Pi_{\mathrm{x},n}(y^0_{-n+1}))$ allows to compute the initial lifted $z_0$ based on $x_0$, using a much smaller amount of lags (at max $n= n_\mathrm{x}$) as, typically, $n_\mathrm{x} \ll n_\mathrm{z}$. As such, the number of necessary lags $n$ does not depend on the dimensionality of the lifted space, but of the original nonlinear system, which drastically reduces the computational complexity.\par  
It is important to emphasize that the conditions discussed in this subsection guarantee local observability and necessary conditions for the local invertibility of \eqref{eq:first_obs_map_nl_aut} and \eqref{eq:nl_alg_constr}. However, for stronger, global guarantees, \cite{Jose_subnet} describes, albeit for a continuous time systems,  that if $f$ is Lipschitz continuous and $h$ has a finite amount of nondegenerate critical points, then $n = 2n_\mathrm{x} +1$ is sufficient to ensure global existence of the reconstructabiltiy map, which is in line with the results in \cite{Takens}.
%%%%%%%%%%%%%%%%%%%%%%%%%
\section{Identification approach}\label{sec:section_identification} Building on the previously discussed results, this section details the proposed Koopman model identification approach for nonlinear systems driven by an external  input and affected by process  and measurement noise.
\subsection{Data generating system}
We consider the following nonlinear system: 
\begin{subequations}\label{eq:data_gen}
\begin{align}
    x_{k+1}&=f_\mathrm{d}(x_k,u_k,e_k), \label{eq:data_gen_x}\\
    y_k &=h(x_k)+e_k, \label{eq:data_gen_y}
\end{align}
\end{subequations}
with $u_k\in\mathbb{U}\subseteq\mathbb{R}^{n_\mathrm{u}}$ the control input and $x_k\in\mathbb{X}\subseteq \mathbb{R}^{n_\mathrm{x}}$ the state. The signal $e_k$ is the sample-path realization of an \rev{i.i.d.~white} noise process of finite variance, taking values in $\rev{\mathbb{E}\subseteq\mathbb{R}^{n_\mathrm{x}}}$ and being independent from $u$ in the statistical sense. 
The functions  $f_{\mathrm{d}}:\mathbb{X} \times\mathbb{U} \times \mathbb{E}\rightarrow\mathbb{R}^{n_\mathrm{x}}$ and $h:\mathbb{X}\rightarrow\mathbb{R}^{n_\mathrm{y}}$ are the state-transition and output functions, respectively. 
It is assumed that the sets $\mathbb{U}$ and $\mathbb{E}$ are such that $\mathbb{X}$ is forward invariant under $f_\mathrm{d}$ and $0\in\mathbb{X}$. 
The objective is to estimate a Koopman model of 
the deterministic (process) part of \eqref{eq:data_gen}. This is done using an input-output data sequence $\mathcal{D}_N=\{(u_k,y_k)\}^N_{k=0}$ collected from the system \eqref{eq:data_gen}. We define next the model structure that we propose to identify a lifted Koopman form of the system under the control  input $u$ and noise process $e$.
\subsection{Koopman model structure}To analytically derive an equivalent Koopman model of \eqref{eq:data_gen}, we begin by decomposing $f_\mathrm{d}(x_k,u_k,e_k)$ into autonomous, input and noise-related components as follows:
%\begin{equation}
\begin{align}\label{eq:nl_sys_decomp}
	f_\mathrm{d}(x_k,u_k,e_k)&=f_\mathrm{d}(x_k,0,0)+\underbrace{f_\mathrm{d}(x_k,u_k,e_k)-f_\mathrm{d}(x_k,0,0)}_{\tilde{f}_\mathrm{d}(x_k,u_k,e_k)}\\
	&=f_\mathrm{d}(x_k,0,0) + \tilde{f}_\mathrm{d}(x_k,u_k,0)+\underbrace{\tilde{f}_\mathrm{d}(x_k,u_k,e_k)-\tilde{f}_\mathrm{d}(x_k,u_k,0)}_{d(x_k,u_k,e_k)} \notag \\
	&=f(x_k) + g(x_k,u_k) + d(x_k,u_k,e_k) \notag
\end{align}
%\end{equation}
where $g(x_k,0)=0$ and $d(x_k,u_k,0)=0$. This decomposition, which always exists, extends the one discussed in \cite{Aut_Iacob_inputs} and \cite{surana_obs} for the noiseless case. To derive an exact finite dimensional Koopman representation, we make the following assumptions.
\begin{assumption}\label{assumption:exact_embedding_aut}
There exists a finite dimensional dictionary of observables $\Phi:
%\mathbb{R}^{n_\mathrm{x}}
\mathbb{X}\rightarrow\mathbb{R}^{n_\mathrm{f}}$ with $\Phi=[\ \phi_1\ \cdots \ \phi_{n_\mathrm{f}}\ ]^\top$ in an appropriate Banach space $\mathcal{F}$ such that \rev{$\Phi \circ f(\cdot) \in \mathrm{span} \{\Phi\}$, which implies that} 
\begin{equation}\label{eq:eq_from_ass_1}
	 %\qquad \implies \qquad 
     \Phi \circ f(\cdot) = A\Phi (\cdot)\; \text{ with }  
     A\in\mathbb{R}^{n_\mathrm{f}\times n_\mathrm{f}}.
\end{equation}
\end{assumption}
\begin{assumption}\label{assumption:exact_embedding_output_aut}
	The output map $h$ can be exactly represented by the observables $\Phi$ in Assumption \ref{assumption:exact_embedding_aut}, in other words \rev{$h \in \mathrm{span}\{\Phi\}$, which implies that}
	\begin{equation}
		%h \in \mathrm{span}\{\Phi\} \implies 
        \rev{h(\cdot)=C\Phi(\cdot)} \; \text{ with } C\in\mathbb{R}^{n_\mathrm{y} \times n_\mathrm{f}}.
	\end{equation}
\end{assumption}
While Assumption \ref{assumption:exact_embedding_output_aut} can be easily satisfied (for example by including $h$ in the dictionary of observables), Assumption \ref{assumption:exact_embedding_aut} is more challenging due to the finite dimensionality. While there exist methods for exact finite embedding of polynomial systems \revtwo{\cite{Brunton_2016}}, \cite{Iacob_poly_system}, \cite{Iacob_CT_block_embedding}, methods for polyflows \cite{Polyflows}, or results in immersion theory \revtwo{\cite{Levine_86}}, \cite{Wang_Immersion}, the conditions for the existence of an exact embedding of more general classes of nonlinear systems are lacking. Hence, it is currently an open question what are the limitations of Assumption \ref{assumption:exact_embedding_aut}, \rev{especially in terms of the findings in \cite{liu-ozay-sontag}.} Otherwise, \eqref{eq:eq_from_ass_1} only holds in an approximative sense, \revtwo{resulting in a certain level of projection error (see it later in this section)}.

In this work, we \revtwo{will primary rely on}  Assumptions \ref{assumption:exact_embedding_aut} and \ref{assumption:exact_embedding_output_aut} \revtwo{to investigate the estimation error of our introduced approach, while dealing jointly with both projection and estimation error is left for future research. However, we will derive both in the finite and infinite dimensional cases the Koopman form of nonlinear systems with control input and influenced by process and measurement noise and we use the resulting form in our method as the model structure to be identified.} 

For this purpose, we formulate the following theorem.
\begin{theorem} \label{TH1} 
	\revtwo{Given a dictionary of continuously differentiable functions $\Phi(\cdot)$ that satisfy} Assumptions \ref{assumption:exact_embedding_aut} and \ref{assumption:exact_embedding_output_aut} \revtwo{for} the nonlinear system \eqref{eq:data_gen}, \revtwo{then, the dynamics of  \eqref{eq:data_gen}} can be \revtwo{represented in} the form:
	\begin{subequations}\label{eq:koop_model_structure}
	\begin{align}
    z_{k+1}&=Az_k+B(z_k,u_k)u_k + K(z_k,u_k,e_k)e_k\label{eq:koop_model_structure_state} \\
    y_k&=Cz_k +e_k\label{eq:koop_model_structure_out}
\end{align}
\end{subequations}
which is an exact finite dimensional Koopman form with innovation noise structure and  \revtwo{$z_k=\Phi(x_k)\in\mathbb{R}^{n_\mathrm{z}}$ being the lifted state with} $n_\mathrm{z}=n_\mathrm{f}$.
\end{theorem}
\begin{proof}
	The proof is given in \cref{sec:appendix_proof_embedding}.
\end{proof}
While \eqref{eq:koop_model_structure} is an exact embedding if Assumptions \ref{assumption:exact_embedding_aut} and \ref{assumption:exact_embedding_output_aut} hold, in an identification setting, it may be desirable to trade accuracy with simplicity of the model. We can conceptually write the Koopman model to be identified as:
\begin{subequations}\label{eq:koopman_model_structure_selection}
	\begin{align}
	z_{k+1}&=Az_k + \mathcal{B}u_k + \mathcal{K}e_k\\
	y_k &= Cz_k + e_k 
	\end{align}
\end{subequations}
where the input function $\mathcal{B}$ may have different complexities, i.e., $\mathcal{B}\in\{B,\;\sum_{i=1}^{n_\mathrm{z}}B_iz_{k,i}+B_0$,$\;B(z_k),\;B(z_k,u_k)\}$, corresponding to a linear, bilinear, input affine or what we call a general model structure. Similar to the $\mathcal{B}$ matrix, the innovation noise term can be considered with various dependencies: $\mathcal{K}\in\{K,\; \sum_{i=1}^{n_\mathrm{z}}K_i z_{k,i}+K_0,\; K(z_k),$ $K(z_k,u_k),\;K(z_k,u_k,e_k) \}$. Choosing a suitable structural form of $\mathcal{B}$ and $\mathcal{K}$ corresponds to a model structure selection problem as in classical system identification.
\par To use the more complex Koopman models for control (not fully LTI), it is possible to cast the model into a \emph{linear parameter-varying} (LPV) form (see \cite{Aut_Iacob_inputs} for the noiseless case) by introducing a so called scheduling variable $p_k$ that is required to be measurable/observable from the system. For nonlinear systems described by LPV models, there exists a multitude of convex and computationally efficient control synthesis methods where the user can also shape performance and achieve global guarantees of stability \cite{Mohammadpour}. To cast \eqref{eq:koopman_model_structure_selection} into an LPV form, we must first note that, in general, the noise $e_k$ is not directly measurable, but thanks to the innovation noise setting of \eqref{eq:data_gen},  $e_k = y_k - Cz_k$ holds. Then, we can conceptually write the LPV form of the Koopman model \eqref{eq:koopman_model_structure_selection} as:
\begin{subequations}\begin{align}
		z_{k+1}&=Az_k + \mathcal{B}_\mathrm{z}(p_k)u_k + \mathcal{K}_\mathrm{z}(p_k)e_k\\
	y_k &= Cz_k + e_k 
	\end{align}\end{subequations}
where $p_k=\mu(z_k,u_k,y_k)$ is a scheduling map and $\mathcal{B}_\mathrm{z} $ with $\mathcal{K}_\mathrm{z}$  belong to a predefined function class such as affine, polynomial or rational, such that $\mathcal{B}_\mathrm{z} \circ \mu = \mathcal{B}$, $\mathcal{K}_\mathrm{z}\circ\mu = \mathcal{K}$  \cite{Aut_Iacob_inputs}, \cite{Mohammadpour}. Note that the dependencies of $\mu$ are based on the choice of $\mathcal{B}$ and $\mathcal{K}$.
\\
\par
\revtwo{Next, we can can relax the finite dimensionality of the dictionary of observables in Assumptions \ref{assumption:exact_embedding_aut} and \ref{assumption:exact_embedding_output_aut} and extend \cref{TH1} to infinite dimensional liftings with $\{\phi_i\}^\infty_{i=1}$. This is formulated in the following lemma.
}
\revtwo{
\begin{lemma}\label{lem:infinite_per_obs}
    Given an observable $\phi_i:\mathbb{X}\rightarrow \mathbb{R}$, with $\phi_i\in\mathcal{C}^1$, and the Koopman operator $\Koop$ acting on $\phi_i$ such that \eqref{eq:koop_composition} holds, then the dynamics of each observable associated with \eqref{eq:data_gen_x} is exactly given by:
    \begin{equation}\label{eq:per_obs_dyna_calc_input_noise}
\begin{split}
			\phi_i(x_{k+1})=\Koop\phi_i(x_k)&+\underbrace{\left(\int^1_0\frac{\partial \phi_i}{\partial x}(f(x_k)+g(x_k,u_k))\dif \lambda\right)g(x_k,u_k)}_{\tilde{b}_\mathrm{x}(x_k,u_k)}
            \\ &+\underbrace{\left(\int^1_0\frac{\partial \phi_i}{\partial x}(f(x_k)+g(x_k,u_k) + \lambda d(x_k,u_k,e_k))\dif \lambda\right)d(x_k,u_k,e_k)}_{\tilde{k}_\mathrm{x}(x_k,u_k,e_k)}.
            \end{split}
		\end{equation}
\end{lemma}
\begin{proof} The proof is given in \cref{sec:appendix_inf_embedding}.
\end{proof}
Note that, as $g(x_k,0)=0$ and $d(x_k,u_k,0)=0$, one can apply %Lemma 1 from \cite{Aut_Iacob_inputs} (
the procedure described in \cref{sec:appendix_proof_embedding} to obtain the more familiar form:
\begin{equation}\label{eq:per_obs_dyna_calc_factorized}
    \phi_i(x_{k+1})=\Koop\phi_i(x_k)+b_\mathrm{x}(x_k,u_k)u_k+
    k_\mathrm{x}(x_k,u_k,e_k)e_k,
\end{equation}
where
\begin{equation*}
	b_\mathrm{x}(x_k,u_k) = \int^1_0\frac{\partial \tilde{b}_\mathrm{x}}{\partial u}(x_k,\lambda u_k)\dif \lambda \quad \text{and} \quad k_\mathrm{x}(x_k,u_k,e_k)= \int^1_0\frac{\partial \tilde{k}_\mathrm{x}}{\partial e}(x_k, u_k, \lambda e_k)\dif \lambda.
\end{equation*}
}
\revtwo{
By relaxing the finite-dimensional assumption of the embedding of the autonomous part, \cref{lem:infinite_per_obs} describes an exact operator-based formulation of the input and noise effects in the infinite-dimensional lifted representation, which holds for each observable $\phi_i$. To the authors' knowledge, this is the first description in the infinite-dimensional setting that provides an exact characterization of the dynamics of observables of nonlinear systems in a general form under input and process noise.}
\par
\revtwo{Finally, we can characterize the projection error of using the finite dimensional Koopman form in \cref{TH1}, i.e., when Assumption \ref{assumption:exact_embedding_aut} is violated.}  

\revtwo{
\begin{corollary}\label{lem:corollary_approx}
 Given a set of observables $\{\phi_i:\mathbb{X}\rightarrow \mathbb{R}\}_{i=1}^{\infty}$, with $\phi_i\in\mathcal{C}^1$, and the Koopman operator $\Koop$ acting on $\phi_i$ such that \eqref{eq:koop_composition} holds. For $n_\mathrm{f} \geq 1$, let $\Phi=[\ \phi_1\ \cdots \ \phi_{n_\mathrm{f}}\ ]^\top$. Then,
\begin{equation}\label{eq:autonomous_approx}
    \Phi(f(x_k))=A\Phi(x_k)+R(x_k),
\end{equation}
where $A\in\mathbb{R}^{n_\mathrm{f}\times n_\mathrm{f}}$ is according to \eqref{eq:comp:A} and $R:\mathbb{X}\rightarrow \mathbb{R}^{n_\mathrm{f}}$ is the state-dependent projection error term. Then, a Koopman embedding of \eqref{eq:data_gen_x} is given by:
\begin{equation}\label{eq:koop_approx_lift}   \Phi(x_{k+1})=A\Phi(x_k)+B(\Phi(x_k),u_k)u_k + K(\Phi(x_k),u_k,e_k)e_k+R(x_k).
 \end{equation}
\end{corollary}
\begin{proof}
    The proof directly follows the derivation for \cref{TH1} given in \cref{sec:appendix_proof_embedding}, by using \eqref{eq:autonomous_approx} in \textit{Step 1}. Alternatively, the observables \eqref{eq:per_obs_dyna_calc_intermediary} can be stacked, then  \eqref{eq:autonomous_approx} and the factorization are applied to obtain \eqref{eq:koop_approx_lift}.
\end{proof}
}
\revtwo{
This result indicates that the term $R(x_k)$ propagates through the derivation unaffected by the input and noise contributions, thus fully characterizing the error of the lifted representation \eqref{eq:koop_approx_lift} of \eqref{eq:data_gen_x}. This shows that one can use a method that characterizes the projection error of the autonomous embedding (e.g., see \cite{zhang_qunatitative}) to fully derive an error bound for the input-driven and noise-affected Koopman representation of \eqref{eq:data_gen_x}, if the analytical derivation in \cref{lem:corollary_approx} is used. 
}

\vspace{.5cm}
\subsection{Identification problem} \label{sec:id:prob}
The objective is to introduce a parametrized version of \eqref{eq:koop_model_structure} to learn the underlying dependencies together with a lifting map using ANNs. This means identifying the lifted state $z_k$, the linear maps $A$, $C$, and the nonlinear maps $B$ and $K$. To this end, we introduce an identification cost function that we chose to be the squared prediction error due to its extended use and success in system identification and its close connection to maximum-likelihood estimators under specific settings \cite{book_ljung}, \cite{Ljung2013}. For this, a predictor is needed and we derive it next. As a first step, we exploit $e_k = y_k - Cz_k$ in the assumed innovation form \eqref{eq:data_gen}  and substitute it in \eqref{eq:koop_model_structure_state} to obtain:
\begin{multline}
z_{k+1}=Az_k +B(z_k,u_k)u_k + K(z_k,u_k,y_k-Cz_k)(y_k-Cz_k)\\	=\underbrace{\left(A-\tilde{K}(z_k,u_k,y_k)C\right)z_k + B(z_k,u_k)u_k + \tilde{K}(z_k,u_k,y_k)y_k}_{\mathcal{F}(z_k,u_k,y_k)} 
\end{multline}
with $\tilde{K}(z_k,u_k,y_k):=K(z_k,u_k,y_k-Cz_k)$. Then, iterating \eqref{eq:koop_model_structure_out} forward in time, for $n\geq 1$, we arrive at 
\begin{equation}\label{eq:predictor_steps}
	\begin{split}
		y_k &= Cz_k + e_k \\
		y_{k+1}&=Cz_{k+1}+e_{k+1}
		=C\mathcal{F}(z_k,u_k,y_k)+e_{k+1}\\
		&\  \vdots\\
		y_{k+n}&=C(\circ_n\mathcal{F})(z_k,u^{k+n-1}_k,y^{k+n-1}_k)+e_{k+n}
	\end{split}
\end{equation}
where $u^{k+n-1}_k=[\ u^\top_k\ \cdots\ u^\top_{k+n-1}\ ]^\top$ and $y^{k+n-1}_k$ is similarly defined.  %in the previous subsection. 
In a compact form:
\begin{equation}\label{eq:compact_predictor_form}
y^{k+n}_k=\Gamma_n(z_k,u^{k+n-1}_k,y^{k+n-1}_k)+e^{k+n}_k,
\end{equation}
with $e^{k+n}_k=[\ e^\top_k\ \cdots\ e^\top_{k+n}]^\top$. Based on the i.i.d  assumption \revtwo{on the noise} $e_k$, the conditional expectation of \eqref{eq:compact_predictor_form} w.r.t.~$e$ based on the available input-output data and $z_k$ is:
\begin{equation}\label{eq:expectation_one_step}
%\begin{split}
\hat{y}^{k+n}_k=\mathbb{E}_e \left\{y^{k+n}_k \mid z_k,u^{k+n-1}_k,y^{k+n-1}_k\right\}=\Gamma_n(z_k,u^{k+n-1}_k,y^{k+n-1}_k)
%\end{split}
\end{equation}
which is the one-step-ahead predictor associated with \eqref{eq:koop_model_structure} along the time interval $[k,k+n]$ and with initial condition $z_k$. %\tb{and can also be interpreted as an $n$-step ahead prediction based on past data.} 
This can be computed for the entire data sequence $\mathcal{D}_N$ as $\hat{y}^N_0=\Gamma_N(z_0,u^{N-1}_0,y^{N-1}_0)$ or, for a particular time-moment, as $\hat{y}_k=\gamma_k(z_0,u^{k-1}_0,y^{k-1}_0)$ with $\gamma_k=C(\circ_k\mathcal{F})$. 

As a next step to identify a Koopman embedding of the data-generating system \eqref{eq:data_gen} in the form of \eqref{eq:koop_model_structure}, we introduce a parameterization of \eqref{eq:koop_model_structure} in terms of 
\begin{subequations}\label{eq:koopman_parametrized_model}
\begin{align}
\hat{z}_{k+1}&=A_\theta \hat{z}_k + B_\theta (\hat{z}_k,u_k)u_k+K_\theta(\hat{z}_k,u_k,\hat{e}_k)\hat{e}_k,\\
\hat{y}_k&=C_\theta \hat{z}_k.
\end{align}
\end{subequations}
In \eqref{eq:koopman_parametrized_model}, $\hat{z}$ is the predicted lifted state, $\hat{y}$ is the predicted output, and $\hat{e}$ is the prediction error. While $A_\theta$ and $C_\theta$ are matrices with their elements as parameters, the maps $B_\theta$ and $K_\theta$ are considered with a given choice of complexity: linear, bilinear, input affine, or general nonlinear dependency. In the linear case, $B_\theta$ and $K_\theta$ are also matrices with their elements as parameters, in the bilinear case, the matrices of the bilinear relations are taken as parameters, while, in the input affine and general cases, $B_\theta$ and $K_\theta$  are represented by ANNs with weights and bias terms collected in $\theta$ together with the weights of a linear bypass. 
The collection of all parameters associated with $A_\theta, \ldots, K_\theta$ are collected in $\theta\in\Theta\subseteq \mathbb{R}^{n_\theta}$. The parameterized model structure gives rise to a parametrized predictor $\Gamma_{N,\theta}$, providing the calculation of $\hat{y}_k$ and the prediction error $\hat{e}_k$ over a data set $\mathcal{D}_N=\{(u_k,y_k)\}^N_{k=0}$ of the data-generating system.

To accurately estimate \eqref{eq:koop_model_structure}, we minimize the $\ell_2$ loss of the error between the measured output $y_k$ and predicted output $\hat{y}_k$:
\begin{equation}\label{eq:cost_initial_function}
V^{\text{pred}}_{\mathcal{D}_N}(\theta)=\frac{1}{N+1}\sum^N_{k=0}\|y_k-\hat{y}_k\|^2_2.
\end{equation}
Note that in \eqref{eq:cost_initial_function}, the initial lifted state $z_0$ is unknown and needs to be optimized during the minimization of \eqref{eq:cost_initial_function}. 
The minimum of \eqref{eq:cost_initial_function} will provide a Koopman model with the best one-step-ahead prediction performance. Later we will investigate  how this estimate is related to the true Koopman embedding of the original nonlinear system, if it exists.     

There are two challenges associated with \eqref{eq:cost_initial_function}: (i) estimation of \eqref{eq:koop_model_structure} in this form does not provide a direct characterisation of the observable or a way how the lifted state can be calculated from measurable variables in the original system; (ii) the computational cost of solving the minimisation problem based on \eqref{eq:cost_initial_function} is high in case of large data sets and numerically challenging 
under ANN parametrisation of $B_\theta$ or $K_\theta$, due to vanishing of the gradients during backward / forward propagation.   

%%%%%%%%%%%%%%%%%%%%%%%%%%%%%%%%%%
\subsection{Subspace encoder} \label{sec:subenc}
To overcome problem (i), in this section, the estimation of the lifted state $z_k$ is considered in terms of an encoder.
By exploiting input-output data, we use the reconsutructability concept, discussed in the autonomous case, which we now generalize for \eqref{eq:koop_model_structure}. 
 Starting with observability, the following \rev{equality} holds based on \eqref{eq:compact_predictor_form}:
\begin{equation}\label{eq:nl_alg_cond_with_input}
\underbrace{\begin{bmatrix}
\Gamma_n(z_k,u^{k+n-1}_k,y^{k+n-1}_k)+e^{k+n}_k\\ \Psi(z_k)
\end{bmatrix}}_{\mathcal{O}_{\mathrm{z},n}(z_k,u^{k+n-1}_k,e^{k+n}_k)}=\begin{bmatrix}
y^{k+n}_k \\ 0
\end{bmatrix}
\end{equation}
where, as in the autonomous case, we have the set of nonlinear constraints $\Psi$. For $n\geq 1$, if $\exists (z_*,w_*)\in\mathbb{R}^{n_\mathrm{z}}\times\mathbb{R}^{nn_\mathrm{u}\times (n+1)n_\mathrm{y} \times (n+1)n_\mathrm{y}}$ for which the Jacobian  $\nabla_{(z_\ast,w_\ast)}\mathcal{O}_{\mathrm{z},n}$ has full row rank $n_\mathrm{z}$, then there exist open sets $\mathbb{Z}_0\subseteq\mathbb{R}^{n_\mathrm{z}}$, $\mathbb{U}_0\subseteq\mathbb{R}^{n_\mathrm{u}}$, $\mathbb{Y}_0\subseteq\mathbb{R}^{n_\mathrm{y}}$, $\mathbb{E}_0\subseteq\mathbb{R}^{n_\mathrm{y}}$, corresponding to the neighborhood of $(z_*,w_*)$ for which 
$\mathcal{O}_{\mathrm{z},n}$ is partially invertible and \eqref{eq:koop_model_structure} is locally observable on $(\mathbb{Z}_0,\mathbb{U}_0,\mathbb{Y}_0,\mathbb{E}_0)$, \rev{see \cite{nl_obs_1982}}. Note that if the representation is locally observable, then the above condition is satisfied for any $n \geq  n_\mathrm{z}-1$. By inverting $\mathcal{O}_{\mathrm{z},n}$, we get
\begin{equation}
z_k=\Lambda_{\mathrm{z},n}(u^{k+n-1}_k,y^{k+n}_k,e^{k+n}_k)
\end{equation}
where $\Lambda_{\mathrm{z},n}:\mathbb{U}^n_0\times \mathbb{Y}^{n+1}_0\times \mathbb{E}^{n+1}_0\rightarrow{\mathbb{R}^{n_\mathrm{z}}}$ is the observability map.  To determine $z_k$ based on past input-output data, we derive
%\begin{equation}
\begin{align} \label{rec:map:full}
z_k&=(\circ_n \mathcal{F})(z_{k-n},u^{k-1}_{k-n},y^{k-1}_{k-n})\\
&=(\circ_n\mathcal{F})(\Lambda_{\mathrm{z},n}(u^{k-1}_{k-n},y^k_{k-n},e^k_{k-n}),u^{k-1}_{k-n},y^{k-1}_{k-n}) \notag\\
&:=\Pi_{\mathrm{z},n}(u^{k-1}_{k-n},y^{k}_{k-n},e^k_{k-n}) \notag
\end{align}
%\end{equation}
where $\Pi_{\mathrm{z},n}:\mathbb{U}^n_0\times \mathbb{Y}^{n+1}_0\times \mathbb{E}^{n+1}_0\rightarrow{\mathbb{R}^{n_\mathrm{z}}}$ is the recosntructability map. Note that the noise sequence $e^k_{k-n}$ is not directly available in practice, hence to compute $z_k$ based on \eqref{rec:map:full},
 again we can exploit the i.i.d.~white noise property of $e_k$ to arrive at:
\begin{equation} \label{enc:form:Koopman}
\bar{z}_k = \mathbb{E}_e \left\{ z_k \mid u^{k-1}_{k-n} , y^k_{k-n}\right\}=\bar{\Pi}_{\mathrm{z},n}(u^{k-1}_{k-n},y^k_{k-n}),
\end{equation}
which mapping gives an efficient estimator of $z_k$ in the conditional mean sense based on past data with a given lag $n$. In fact, \eqref{enc:form:Koopman} functions as an encoder, mapping from the past data to the lifted state $z_k$, i.e., a subspace of the lifted state space. However, an exact calculation of the encoder $\bar{\Pi}_{\mathrm{z},n}$ for a given ANN parametrization of $f_\theta$ and $h_\theta$ is infeasible in practice, due to the required analytic inversion of $\mathcal{O}_{\mathrm{z},n}$ to get $\Lambda_{\mathrm{z},n}$ % $\Pi_{\mathrm{z},n}$?} 
and the computation of the conditional expectation of $\bar{\Pi}_{\mathrm{z},n}$ under the unknown probability density function of $e_k$. Hence, our objective is to learn $\bar{\Pi}_{\mathrm{z},n}$ directly from the data by introducing a parametrized function $\bar \Pi^\eta_{\mathrm{z},n}$ with parameters $\eta\in\Upsilon \subseteq\mathbb{R}^{n_\eta}$, e.g., using an ANN in the general case, which is co-estimated with $A_\theta$, $B_\theta$, $C_\theta$ and $K_\theta$\footnote{\rev{We note an alternate method described in \cite{otto_bilinear}, where the authors describe the Koopman representation as a hidden Markov model and use Kalman filtering to estimate the latent (lifted) state based on partial observations. However, the bilinear form and the way noise affects the lifted model are chosen empirically, whereas \eqref{eq:koopman_parametrized_model} \revtwo{gives} an analytic form.}}. 
\par Note that, similar to the autonomous case discussed in Section \ref{sec:obs_and_reconstr_aut}, we can conceptually show that exploiting the observability and reconstructability properties of the nonlinear system \eqref{eq:data_gen}, the potentially needed number of lags $n\geq n_\mathrm{z}-1$ is greatly reduced. Given a lifting map $\Phi:\mathbb{R}^{n_\mathrm{x}}\rightarrow\mathbb{R}^{n_\mathrm{z}}$, the state $z_k=\Phi(x_k)$ can be calculated based on the reconstructed $x_k$ using a number of lags $n\geq n_\mathrm{x}-1$, as detailed in \cite{Gerben_Aut}, for the reconstructabillity map associated \revtwo{with} \eqref{eq:data_gen}.  As such, when computing $\bar \Pi^\eta_{\mathrm{z},n}$, the number of lags needed to estimate $z_k$ is related to the dimension of the underlying nonlinear system, rather than that of the lifted system.  It is important to note that, while $n_\mathrm{x}-1$ guarantees local invertibility, for global observability guarantees one would need to increase the number of lags $n$. For example, works such as \cite{Stark_1} and \cite{Stark_2} describe a sufficient condition for reconstruction of $x_k$ to be $n\geq 2n_\mathrm{x}$, for nonlinear systems with deterministic and stochastic forcing, respectively. Note that, as the last step in \eqref{eq:predictor_steps} is $k+n$ instead of $k+n-1$, we subtract 1 from the given value in \cite{Stark_1} and \cite{Stark_2}, which is $2n_\mathrm{x}+1$. While \eqref{eq:nl_alg_cond_with_input} is more complex than the maps discussed in \cite{Stark_1} and \cite{Stark_2}, these results can still serve as a guideline when performing experiments and choosing $n$. Furthermore, if $n_\mathrm{x}$ is known, although $\mathcal{O}_{\mathrm{z},n}$ may be locally invertible for smaller values of $n$, it is still a safe choice to start with $n\geq n_\mathrm{x}-1$, which provides local guarantees for reconstructability.

\subsection{Model estimation via multiple shooting and subspace encoding}

To overcome problem (ii), we truncate the $\ell_2$ prediction loss \eqref{eq:cost_initial_function} to a horizon of length $T$ and we divide the data into subsections on which the truncated prediction loss is calculated, giving a so called \emph{multiple shooting} form of the optimization problem. This approach reduces the computational cost and improves numerical stability \cite{Ribeiro}\revtwo{, and contains \emph{single shooting} (corresponding to the minimization of \eqref{eq:cost_initial_function}) as a special case}. Accordingly, the prediction loss is reformulated as:  
\begin{subequations}\label{eq:prediction_loss_final}
\begin{align}
V^{\mathrm{enc}}_{\mathcal{D}_N}(\theta,\eta)&=\frac{1}{N_\mathrm{sec}}\sum^{N-T+1}_{k=n+1}\sum^{T-1}_{\tau=0}\|y_{k+\tau}-\hat{y}_{k+\tau|k}\|^2_2 \label{eq:pred:enc}\\
\hat{z}_{k|k}&=\bar \Pi^\eta_{\mathrm{z},n}(u^{k-1}_{k-n},y^k_{k-n})\\
\hat{z}_{k+\tau+1|k}&=A_\theta \hat{z}_{k+\tau|k}+B_\theta(\hat{z}_{k+\tau|k},u_{k+\tau})u_{k+\tau}  %\\
%&\qquad\qquad\quad\;
+ K_\theta(\hat{z}_{k+\tau|k},u_{k+\tau},\hat{e}_{k+\tau|k})\hat{e}_{k+\tau|k}\nonumber \\
\hat{y}_{k+\tau|k}&=C_\theta \hat{z}_{k+\tau|k}\\
\hat{e}_{k+\tau|k}&=y_{k+\tau}-\hat{y}_{k+\tau|k}
\end{align}
\end{subequations}
with $N_\mathrm{sec}=(N-T-n+1)T$. Here, the notation $(|)$ is introduced to make the distinction $(\text{current time index }|\text{ start index})$ of variables associated with a given section of the data. Note that chopping up the cost function to $T$-length sections would require the introduction of the initial condition  $\hat{z}_{k|k}$ as an optimization variable, which would tremendously increase the number of optimization variables, potentially losing any computational benefit of the multiple-shooting-based reformulation. 
To avoid this, the previously introduced subspace encoder is used 
\begin{equation} \label{subenc:2}
\hat{z}_{k|k}:=\bar \Pi^\eta_{\mathrm{z},n}(u^{k-1}_{k-n},y^k_{k-n}),
\end{equation}
with $\eta\in\Upsilon\subseteq\mathbb{R}^{n_\eta}$, corresponding to a general ANN parameterization of $\bar\Pi^\eta_{\mathrm{z},n}$ under $n$ number of lags.

Now we can co-estimate the encoder $\bar \Pi^\eta_{\mathrm{z},n}$ together with the parameterized ~matrices $(A_\theta,C_\theta)$ and matrix functions $(B_\theta,K_\theta)$ of the Koopman model. Fig.~\ref{fig:encoder_structure} shows the resulting network structure. Note that the computational cost of \eqref{eq:prediction_loss_final} can be further reduced by using a batched formulation, which allows to compute the cost of each section in parallel, independent from each other. This is achieved by only summing over a subset of the sections, which can also partially overlap. The reformulated batch cost function is:
\begin{subequations}
	\begin{align}
		V^{\mathrm{batch}}_{\mathcal{D}_N}&(\theta,\eta)=\frac{1}{N_{\text{batch}}}\sum_{k\in\mathcal{I}}\frac{1}{T}\sum^{T-1}_{\tau=0}\|y_{k+\tau}-\hat{y}_{k+\tau|k}\|^2_2\\
		&\mathcal{I} \subset\mathbb{I}^{N-T+1}_{n+1}=\{n+1,n+2,\dots,N-T+1\}\\
		&\text{s.t. } |\mathcal{I}| = N_{\text{batch}}
	\end{align}
\end{subequations}
which enables the application of advanced batch optimization algorithms like Adam \cite{adam}. Moreover, the complete dataset does not need to be fully loaded into the memory, making the implementation more efficient \cite{Gerben_Aut}.
\begin{figure}[t!]
  \centering
\includegraphics[scale=0.75]{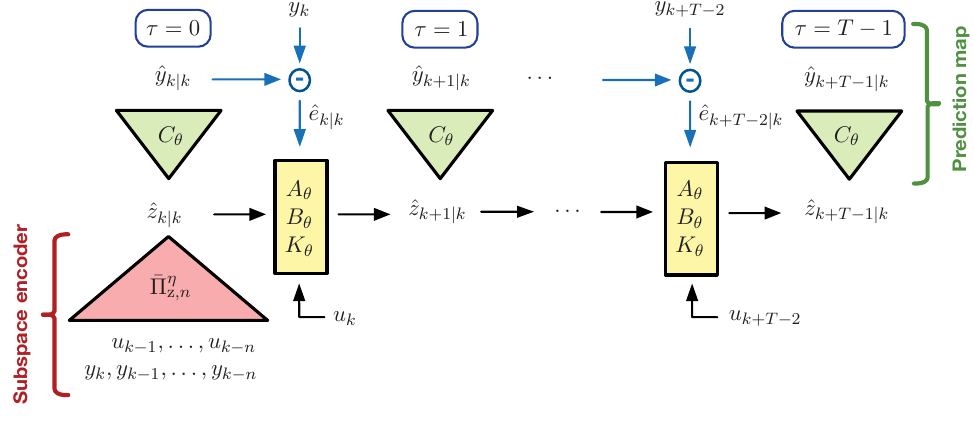} 
  \caption{Network architecture. The lifted state at moment $k$, i.e., $\hat{z}_{k|k}$, is estimated using the encoder function $\bar \Pi^\eta_{\mathrm{z},n}$ based on previously measured input and output data.} \vspace{-.5cm}\label{fig:encoder_structure}
\end{figure}
%%%%%%%%%%%%%%%%%%%%%%%%%%%%%%%%%%%%%%%%%
\section{Consistency analysis} \label{sec:convergence_and_consistency}

Next, we show the consistency of the proposed identification scheme that corresponds to the minimization of \eqref{eq:prediction_loss_final}. In fact, under the assumption that an exact Koopman embedding of \eqref{eq:data_gen} in the form of \eqref{eq:koop_model_structure} exists, 
we show that the resulting model estimate will converge to an equivalent representation of the system in the form \revtwo{of} \eqref{eq:koop_model_structure}, \revtwo{i.e., the estimation error will converge to zero with probability 1,} if the number of data points in the available data set $\mathcal{D}_N$ tends to infinity, that is, $N\rightarrow \infty$. \revtwo{For this analysis, we will first rely on the assumption of a finite-dimensional embedding and then discuss the conditions under which the results generalize to the infinite dimensional case. While we will not explicitly treat the projection error of inexact finite dimensional embedding, combining the current results with existing projection error characterization methods through \cref{lem:corollary_approx} can be used to derive joint error bounds.} The consistency analysis discussed in this section is an adaptation of the arguments in \cite{Gerben_Aut} to the considered Koopman identification problem.

\subsection{Convergence}
As a first step, the convergence of the Koopman model estimate will be shown. By satisfying Assumptions \ref{assumption:exact_embedding_aut} and \ref{assumption:exact_embedding_output_aut}, the data-generating system \eqref{eq:data_gen} has an exact representation by the Koopman form \eqref{eq:koop_model_structure} according to Theorem \ref{TH1}. To show convergence, this equivalent Koopman form of the system needs to satisfy the following stability \revtwo{assumption}:
\begin{assumption} %[Stability of the data-generating system]
\label{assum:S1}
    %The Koopman form \eqref{eq:koop_model_structure} of the data-generating system 
    %has the property that, f
    For any $\delta> 0$, there exist a $c\in [0, \infty)$ and a $\lambda \in[0,1)$ such that
    \begin{align}
        \mathbb{E}_e \{ \|y_k - \tilde{y}_k\|_2^4 \} < c  \lambda^{k}, \quad \forall k\geq 0,
    \end{align}
    under any initial conditions $ z_0,\tilde{z}_0\in\mathbb{R}^{n_\mathrm{z}}$  with $\|z_0-\tilde{z}_0\|_2<\delta$ and $\{(u_\tau,e_\tau)\}_{\tau={0}}^k\in \mathcal{S}_{[0,\infty]}$, where $\mathcal{S}_{[0,\infty]}$ is the $\sigma$-algebra generated by the random variables $\{(u_\tau,e_\tau)\}_{\tau={0}}^\infty$, and %the random variables 
    $y_k$ and $\tilde{y}_k$ satisfy \eqref{eq:koop_model_structure}  with the same $(u_k,e_k)$, but with $z_0$ and $\tilde{z}_0$.
\end{assumption}

\revtwo{The above assumption  simply means that the data-generating system \eqref{eq:data_gen}, i.e., its Koopman form \eqref{eq:koop_model_structure}, is stable in an incremental sense. This means that starting from two distinct initial conditions, the solutions converge to each other under the same input $u$ and noise process $e$. If the data-generating system is not stable, then in practice, data cannot be obtained under free choice of the excitation $u$. In order to conduct experiments, a controller is required to be used to run the data-collection in closed loop. However, that will introduce a correlation between $u$ and the noise $e$, which will potentially require a different estimation approach and also stochastic analysis to show consistency. Hence, \cref{assum:S1} is quite commonly taken in the analysis of system identification methods.}
\par

To identify \eqref{eq:koop_model_structure}, the \emph{model structure} $M_\xi$ is composed of the Koopman model \eqref{eq:koopman_parametrized_model} with the forms of parametrization discussed in Section \ref{sec:id:prob} and the subspace encoder \eqref{subenc:2} with the parametrization discussed in Section \ref{sec:subenc}, giving the total parameter vector
$\xi = [\begin{array}{cc} \theta^\top & \eta^\top \end{array}]^\top$ that is restricted to vary in a compact set $\Xi\subset \mathbb{R}^{n_\xi}$.
This gives the \emph{model set} $\mathcal{M}=\{ M_\xi \mid \xi\in\Xi\} $. For each $\xi \in \Xi$, the model $M_\xi$ with a given encoder lag $n\geq 1$, can be written in a \emph{one-step-ahead predictor} form 
\begin{equation} \label{eq:predic:comp}
    \hat{y}_{k+\tau|k} =  \gamma^\mathrm{pred}_\tau (y^{k+\tau-1}_{k-n},u^{k+\tau-1}_{k-n},\xi),
\end{equation}
which is a combination of $\gamma_k$ based on \eqref{eq:expectation_one_step} and encoder \eqref{enc:form:Koopman}. Note that based on the parametrizations discussed in Sections \ref{sec:id:prob} and \ref{sec:subenc}, 
$\gamma^\mathrm{pred}_{(\cdot)}$ is differentiable w.r.t.~$\xi$ everywhere on an open neighborhood $\breve{\Xi}$ of $\Xi$. Furthermore, \eqref{eq:predic:comp} is required to be stable under any perturbation of the measured data, which is expressed as follows.
\begin{assumption}
\label{assum:M1} There exist scalars $c\in [0, \infty)$ and $\lambda \in[0,1)$ such that,
for any $\xi\in\breve\Xi$ and for any $\{(y_\tau,u_\tau)\}_{\tau=-n}^{k},\{(\tilde{y}_\tau,\tilde{u}_\tau)\}_{\tau=-n}^{k} \in\mathbb{R}^{(n_\mathrm{y}+n_\mathrm{u})\times (n+k+1)}$,
the predictors
  \begin{equation*} \hat{y}^\mathrm{pred}_k = \gamma^\mathrm{pred}_k (y^{k-1}_{-n},u^{k-1}_{-n},\xi), \quad 
  \tilde{y}^\mathrm{pred}_k = \gamma^\mathrm{pred}_k (\tilde y^{k-1}_{-n},\tilde u^{k-1}_{-n},\xi),
  \end{equation*} 
satisfy
    \begin{equation} \label{eq:M1}
        \|\hat{y}^\mathrm{pred}_k-\tilde{y}^\mathrm{pred}_k\|_2   \leq c \sigma(k) ,  \quad \forall k\geq 0,
\end{equation}
with $\sigma(k)=\sum_{\tau=-n}^{k} \lambda^{k-\tau} \left(\|{u}_\tau-\tilde{u}_\tau\|_2 + \|{y}_\tau-\tilde{y}_\tau\|_2 \right)$
 and
 $\|\gamma^\mathrm{pred}_k(0^{k-1}_{-n},0^{k-1}_{-n},\xi)\|_2 \leq c$, with $0^{k-1}_{-n}= [\ 0\ \cdots\ 0\ ]^\top$.
%    \begin{align}
%        \|\gamma^\mathrm{pred}_k(\{0\}_{\ell=0}^{k-1},\{0\}_{\ell=0}^k,\theta, 0)\|_2 \leq c,
%    \end{align}
Additionally, there exist $c\in [0, \infty)$ and $\lambda \in[0,1)$ such that $\frac{\partial}{\partial \xi}\gamma^\mathrm{pred}_k$  satisfies \eqref{eq:M1} as well. 
\end{assumption}

\revtwo{ The above given assumption states that any $M_\xi\in \mathcal{M}$ is  stable in the sense that bounded deviations of input and output trajectories in two data sets $\mathcal{D}_N$ and $\tilde{\mathcal{D}}_N$ lead to bounded differences between the predicted responses by the model $M_\xi$. This commonly taken technical assumption is required to establish asymptotic results by avoiding that models in $\mathcal{M}$ correspond to unstable one-step-ahead predictors, which in turn lead to an explosion of the estimation error as $N \rightarrow \infty$.   
} \par

Now we can state the following theorem on convergence of the estimator.

\begin{theorem} %[Convergence] 
\label{lem:convergence}
If the Koopman form \eqref{eq:koop_model_structure} of the data-generating system satisfies \revtwo{Assumption} \ref{assum:S1} with a quasi-stationary $u$ independent of the white noise process ${e}$ and the model set $\mathcal{M}$ defined by  \eqref{eq:koopman_parametrized_model} and \eqref{subenc:2} satisfies \revtwo{Assumption} \ref{assum:M1}, then
    \begin{align} \label{eq:conv}
        \underset{\mathrm{vec}(\theta,\eta)\in\Xi}{\mathrm{sup}} \left \|V^{\mathrm{enc}}_{\mathcal{D}_N}(\theta,\eta) - \mathbb{E}_e\{ V^{\mathrm{enc}}_{\mathcal{D}_N}(\theta,\eta)\}\right\|_2 \rightarrow 0,
    \end{align}
    with probability 1 as $N \rightarrow \infty$. 
\end{theorem}

\begin{proof}
As the identification criterion \eqref{eq:pred:enc} satisfies Condition C1 in \cite{ljung1978convergence}, the proof of \cite[Lemma 3.1]{ljung1978convergence} applies to the case considered.  \end{proof}

\subsection{Consistency} \label{sec:consistency}
To formally show consistency, the Koopman form \eqref{eq:koop_model_structure} of the data-generating system must belong to the chosen set of models $\mathcal{M}$. This means that there exists a $\xi_\mathrm{o} \in \Xi$ such that the one-step-ahead predictor $\gamma^\mathrm{pred}_k$ associated with $M_{\xi_\mathrm{o}}$ and the Koopman form \eqref{eq:koop_model_structure} of the data-generating system are the same. Unfortunately, a system can have many equivalent state-space representations; hence, even if the estimator converges in terms of \eqref{eq:conv}, it can do so to just a $\xi$ that corresponds to a different state-space representation expressing the same dynamics. Therefore, we need to understand consistency w.r.t.~an equivalence class of \eqref{eq:koop_model_structure}. For this purpose, introduce $\Xi_\mathrm{o} \subset \Xi$, which contains all $\xi_\mathrm{o} \in \Xi_\mathrm{o}$ that correspond to equivalent models of the data-generating system in the one-step-ahead prediction sense. Note that if $\Xi_\mathrm{o} = \varnothing$, then chosen parameterization based  $\mathcal{M}$ can not describe the true \eqref{eq:koop_model_structure} and consistency cannot hold.

Before arriving at our result, we need to make sure that the data contains enough information to recover the true underlying dynamics:
\begin{condition} %[Persistence of excitation]
    \label{assum:PE} For the given model set $\mathcal{M}=\{ M_\xi \mid \xi\in\Xi\} $ with $\xi = [\begin{array}{cc} \theta^\top & \eta^\top \end{array}]^\top$, we call the input sequence $\{u_k\}_{k=0}^{N-1}$ in $\mathcal{D}_{N}$ generated by the Koopman form \eqref{eq:koop_model_structure} of the data-generating system {\it weakly persistently exciting}\footnote{\rev{Note that the notion of persistency of excitation used here is in line with the classical notion of informativity in system identification, see \cite{book_ljung}, and it implies distinguishability (under the data $\mathcal{D}_{N}$) of the achieved cost and the predictors corresponding to the used model structure  for \revtwo{$\xi\in\Xi$} values that do not correspond to equivalent models.}}, if for all pairs of parameters given by $\xi_1\in\Xi$ and $\xi_2\in\Xi$ for which the function mapping is unequal, i.e., $V_{(\cdot)}^\mathrm{enc}(\theta_1,\eta_1) \neq V_{(\cdot)}^\mathrm{enc}(\theta_2,\eta_2)$, we have
    \begin{align}
        V_{\mathcal{D}_{N}}^\mathrm{enc}(\theta_1,\eta_1) \neq V_{\mathcal{D}_{N}}^\mathrm{enc}(\theta_2,\eta_2)
    \end{align}
    with probability 1.
\end{condition}

Next, to prove consistency, all elements of $\Xi_\mathrm{o}$ must have minimal asymptotic cost in terms of $\lim_{N\rightarrow \infty} V_{\mathcal{D}_N}^\mathrm{enc}(\theta,\eta)$. However, due to the prediction error nature of the used $\ell_2$-type loss function together with the existence of $\mathbb{E}_e \{ V_{\mathcal{D}_{N}}^\mathrm{end}(\theta,\eta) \}$ (shown in Theorem \ref{lem:convergence}), the minimal asymptotic cost property of $\Xi_\mathrm{o}$ is satisfied by $V_{\mathcal{D}_N}^\mathrm{enc}$. For a detailed proof, see \cite{ljung1978convergence}. 

\begin{theorem} %[Consistency] 
\label{lem:consistency}
Under the conditions of Theorem \ref{lem:convergence} and Condition \ref{assum:PE}, 
    \begin{align}
        \underset{N \rightarrow \infty}{\lim} \hat{\xi}_N \in \Xi_\mathrm{o}
    \end{align}
    with probability 1, where
    \begin{align}
        \hat{\xi}_N =\underset{\mathrm{vec}(\theta,\eta)\in\Xi} { \arg \min} %\underset{\theta}{\argmin} 
        V_{\mathcal{D}_{N}}^\mathrm{enc}(\theta, \eta).
    \end{align}
\end{theorem}
\begin{proof}
The proof is a direct application of Lemma 4.1 in \cite{ljung1978convergence} because the loss function \eqref{eq:pred:enc} fulfills Condition (4.4) in \cite{ljung1978convergence}. 
\end{proof}

\revtwo{Note that the provided consistency proof assumes a finite-dimensional Koopman system representation that has a finite, but possibly arbitrarily large, lifted state dimension $n_\mathrm{z}$. It is nevertheless possible to extend this reasoning for the limit of $n_\mathrm{z}$ growing to infinity. However, as the number of parameters $n_\theta$ %will generally
grow for a growing number of lifted states, the data length $N$ needs to grow with a faster rate than the number of parameters $n_\theta$, i.e. the ratio $\frac{N}{n_\theta}$ needs to tend infinity. Examples of such a setting can for instance be found in the identification of \emph{frequency response functions} (FRF), where FRFs with infinitesimally small resolution, and hence, infinite number of parameters, can be identified as long as both the number of data points per period, and the number of periods are growing towards infinity \cite{Pintelon_12}}.

\subsection{Discussion}
\revtwo{The consistency result essentially proves that, when the data set grows to infinity ($N\rightarrow\infty$), the identified model converges to an equivalent \revtwo{Koopman} representation \eqref{eq:koop_model_structure} of the original nonlinear system \eqref{eq:data_gen} with probability 1. \revtwo{This means that the estimation error of the Koopman model tends to zero asymptotically as $N\rightarrow \infty$.} This holds for representations with arbitrarily large lifting dimension $n_\mathrm{z}$ that are assumed to introduce no projection error. In the literature, such an assumption on the projection error is typically used. For example, in the autonomous case, multiple works tackle estimation error, usually either by directly investigating the projected operator \cite{finite_data_err_bounds}, working in the infinite-dimensional setting of the RKHS or approximating the action of the Koopman operator via a finite Mercer series expansion \cite{kernel_bounds}, or by assuming invariance via specific kernel choices \cite{kohne_kedmd_bound}. For systems with input, the lifted representation is most often assumed to be bilinear \cite{finite_data_err_bounds}, \cite{schaller_proj_bounds}, \cite{straesser_bilinear_err_bounds_ieee}, with full state availability. Also, noise is usually lacking in the works that investigate error bounds, with the exception of simplistic noise assumptions (e.g. Brownian motion in stochastic differential equations \cite{finite_data_err_bounds}), or only measurement noise in the autonomous case in the infinite data and infinite dictionary limits \cite{llamazareselias}.  In comparison, this work introduces an identification method to estimate Koopman models of nonlinear systems driven by external input and affected by process and measurement noise, based on analytically correct model structures with general complexity and under the practical setting of using input-output data only without any direct state measurement. To the authors' knowledge, the consistency result in \Cref{sec:convergence_and_consistency} provides the first proof of the estimation error of the Koopman model in the presence of both input, process, and measurement noise, under the assumption of no projection error. An exciting future direction in this respect is to extend the current results to formulate finite sample bounds, e.g., based on the results of \cite{1024346}.}
\par

\revtwo{While in this paper we do not explicitly treat the projection error of inexact finite dimensional embedding, combining the current consistency result described in \Cref{sec:convergence_and_consistency} with existing methods
that quantify the \revtwo{projection} error when Assumption \ref{assumption:exact_embedding_aut} does not hold can be used to derive joint error bounds 
through \cref{lem:corollary_approx}.
For example, characterization or mitigation of the $n$-step reconstructabiltiy error could be achieved via a reprojection approach as discussed in \cite{Pvangoor}. Alternatively, \cref{lem:corollary_approx} shows that one could possibly use an identification method that characterizes the approximation error of the autonomous embedding, including both estimation and projection errors (or circumventing the latter via specific kernel choices) \cite{kohne_kedmd_bound,Parachuri, kernel_bounds, schaller_proj_bounds,straesser_bilinear_err_bounds_ieee,yadav_mauroy_berstein}, to fully derive an error bound for the input-driven and noise-affected Koopman representation of \eqref{eq:data_gen_x}, when the analytical derivation in \cref{lem:corollary_approx} is used. This would allow for a much simpler result than in \cite{straesser_bilinear_err_bounds_ieee}, where, for a bilinear system with no noise, the residual depends on both state and input. } 

%%%%%%%%%%%%%%%%%%%%%%%%%%%%%%%%%%%%%%%%%
\section{Experiments and results}\label{section_experiments} 
Next we test the proposed Koopman model identification approach on a simulation study of a nonlinear Wiener-Hammerstein system that admits an exact Koopman embedding and the publicly available Bouc-Wen oscillator-based identification benchmark
that has hysteretic behavior and it is notoriously hard to identify. Finally, we apply our approach to capture a Koopman form of the the flight dynamics of a Crazyflie 2.1 nano quadcopter using measured flight data.

\subsection{Wiener-Hammerstein system}\label{sec:wiener_hammer}
To illustrate the performance of the proposed Koopman model structure and learning method as well as the ability to handle process noise, we consider a Wiener-Hammerstein system described by the interconnection of 2 \textit{single-input single-output} (SISO) LTI blocks and a polynomial nonlinearity, as shown in Fig.~\ref{fig:block_figure}. The dynamics of the first block are
\begin{figure}[!tb]
    \centering
    \begin{minipage}{.8\textwidth}
        \centering
        \includegraphics[width=0.95\textwidth]{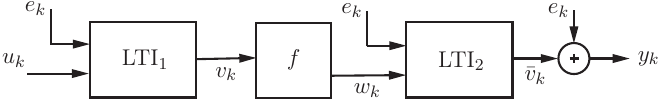}
        \caption{Wiener-Hammerstein system}\vspace{-.5cm}
        \label{fig:block_figure}
    \end{minipage}%
\end{figure}
\begin{subequations}\label{eq:lti_block_dyn_1}
	\begin{align}
		x_{k+1}&=A_1x_k+B_1u_k+K_1e_k \\
v_k &= C_1x_k
	\end{align}
\end{subequations}
with $x_k\in\mathbb{R}^{n_\mathrm{x}}$ the state vector, $u_k\in\mathbb{R}$ the input, and $e_k\sim \mathcal{N}(0,\sigma_\mathrm{e}^2)$ being an i.i.d.~white noise process with standard deviation $\sigma_\mathrm{e}>0$, while $A_1\in\mathbb{R}^{n_\mathrm{x}\times n_\mathrm{x}}$, $B_1\in\mathbb{R}^{n_\mathrm{x}}$, $K_1\in\mathbb{R}^{n_\mathrm{x}}$ and $C_1\in\mathbb{R}^{1\times n_\mathrm{x}}$. The output $v_k\in\mathbb{R}$ of \eqref{eq:lti_block_dyn_1} is affected by %polynomial nonlinearity 
\begin{equation} \label{eq:nl}
	w_k = f(v_k)=\alpha_0 + \alpha_1v_k + \alpha_2 v^2_k + \alpha_3  v^3_k,
\end{equation}
with $\{\alpha_i\}_{i=1}^3\subset\mathbb{R}$. As $v_k = C_1x_k$, \eqref{eq:nl} can be written as:
\begin{equation}
f(v_k)=f(C_1x_k)=\alpha_0+ \alpha_1C_1x_k + \alpha_2 C^{(2)}_1x_k^{(2)} + \alpha_3C^{(3)}_1x_k^{(3)}.
\end{equation}
We denote by $^{(i)}$ the Kronecker power, i.e., $C^{(3)}=C_1 \otimes C_1 \otimes C_1$, where $\otimes$ is the Kronecker product. Finally, the second linear block is described as 
\begin{subequations}\label{eq:lti_block_dyn_2}
\begin{align}
\bar{x}_{k+1} &= A_2 \bar{x}_k + B_2w_k + K_2 e_k\\
y_k &= \underbrace{C_2\bar{x}_k}_{\bar{v}_k} + e_k
\end{align}
\end{subequations}
with $\bar{x}_k\in\mathbb{R}^{n_\mathrm{\bar{x}}}$ the state vector and with  matrices similarly defined as for the first LTI block. With $n_\mathrm{x}=n_{\bar{\mathrm{x}}}=2$, the exact numerical values of the matrices and the polynomial coefficients are given in \cite{Iacob_DT_WH_embedding}. The system described by \eqref{eq:lti_block_dyn_1}--\eqref{eq:lti_block_dyn_2} can be exactly represented by a finite dimensional Koopman model \eqref{eq:koop_model_structure}. For brevity, we refer the reader to \cite{Iacob_DT_WH_embedding} for the derivations,  which uses a similar finite dimensional conversion approach to \cite{Iacob_CT_block_embedding}. The resulting lifted state \revtwo{of the exact Koopman model has dimension $n_\mathrm{z}=12$}.

To generate data, the input is considered as a white noise process with uniform distribution $u_k\sim \mathcal{U}(-1,1)$, \rev{independent} of $e$, while the standard deviation $\sigma_\mathrm{e}$  of $e$ is chosen to obtain $5-30$ dB levels of \emph{signal-to-noise ratio} (SNR) at the output. Based on this,  train, validation and test data sets are generated of length $N=12000$, 4000 and 4000, respectively, with independent realizations of $u$ and $e$. 

\begin{figure}[!htb]
    \centering
    \begin{minipage}{.85\textwidth}
    \centering
        \includegraphics[width=0.9\textwidth]{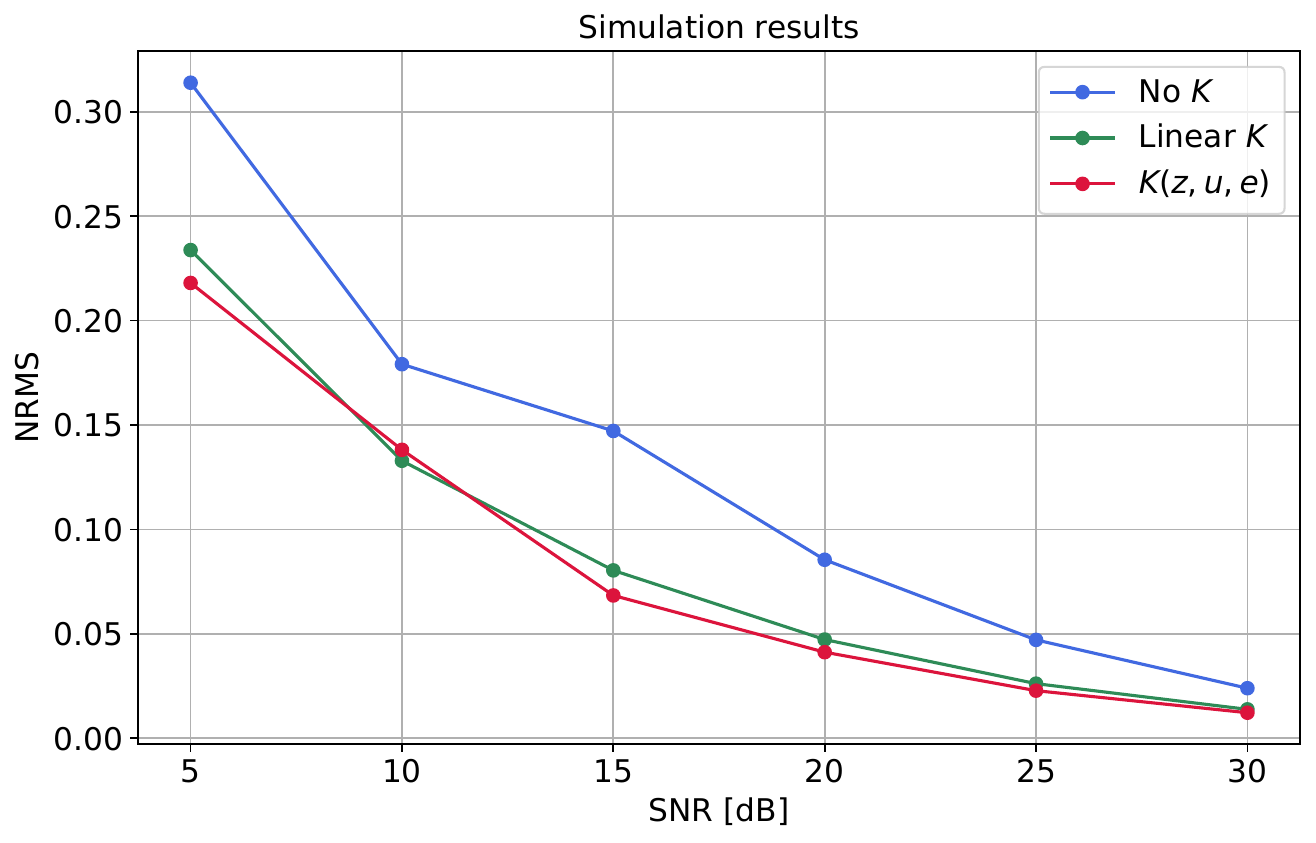}
        \caption{NRMS of the simulation responses of the process part of the Koopman models  w.r.t.~a noiseless test data set, when the Koopman models are estimated with noisy data under various SNR levels (Wiener-Hammerstein system). } 
        \label{fig:WH_simulation_results}
    \end{minipage}% 
    \vspace{-4mm}
\end{figure}

\begin{figure}[!htb]
	\centering
	\begin{minipage}{0.85\textwidth}
        \centering
        \includegraphics[width=0.9\textwidth]{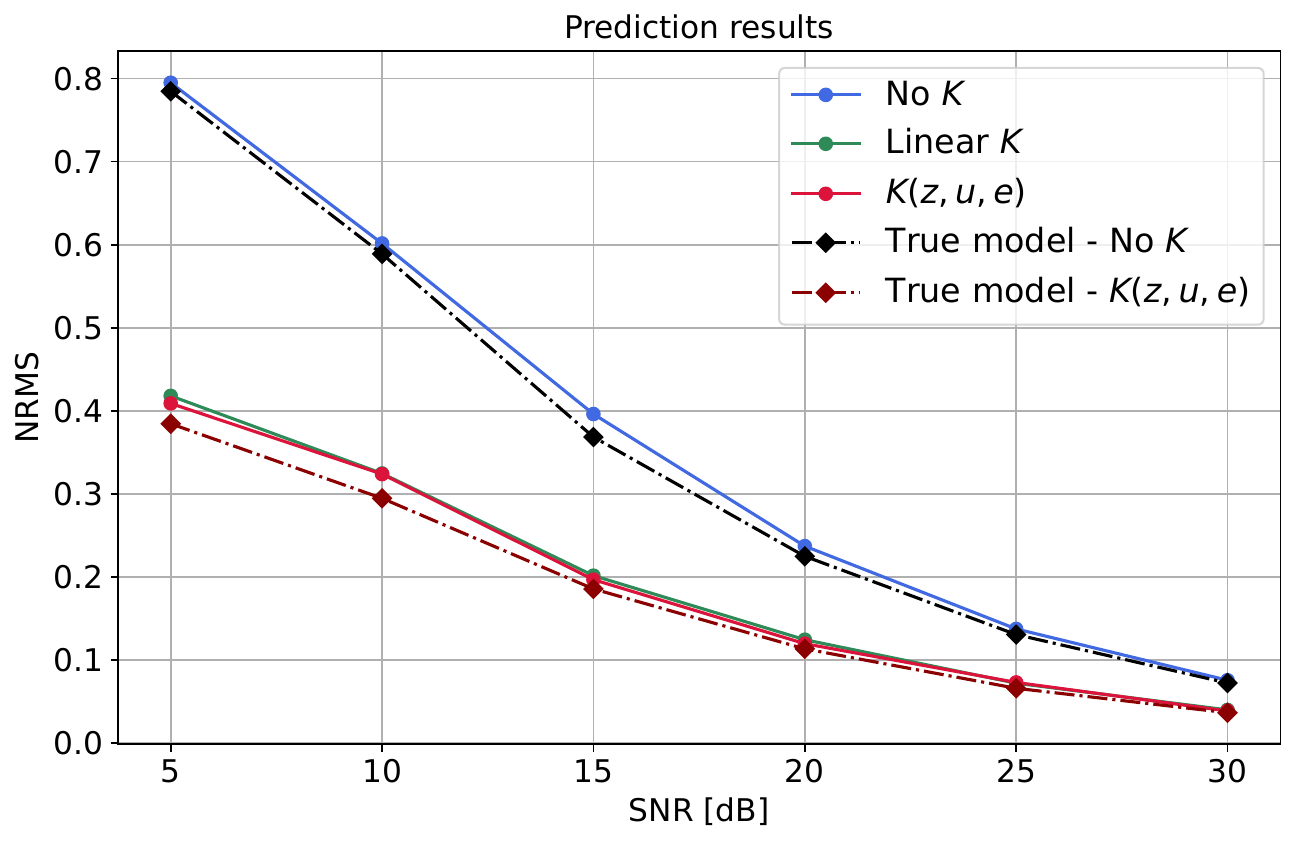}
        \caption{NRMS of the one-step-ahead prediction by the Koopman models w.r.t. noisy test data sets with different SNRs, when the Koopman models are also estimated with noisy data under these SNR levels (Wiener-Hammerstein system).}
        \label{fig:WH_filtering_results}
    \end{minipage}
    \vspace{-4mm}
\end{figure}

In the considered Koopman model structure $M_\xi$, defined by \eqref{eq:koopman_parametrized_model} and \eqref{subenc:2}, 
the encoder $\bar \Pi^\eta_{\mathrm{z},n}$, the input function $B_\theta$ and the innovation noise structure $K_\theta$ are parametrized as feedforward neural networks with 1 hidden layer, tanh activation and 40 neurons per layer for the encoder and $B_\theta$ function, and 80 neurons for $K_\theta$, while $A_\theta$ and $C_\theta$ are taken as parametrized matrices (also $K_\theta$ in the linear case). The parameters are initialized using Xavier initialization and we employ early stopping. The lifting dimension \revtwo{model structure} is selected to coincide with the exact \revtwo{Koopman embedding}, i.e., $n_z=12$ and we use an encoder lag of $n=12$. The prediction horizon is chosen to be  $T=51$ with a batch size of 256. For training, Adam optimization \cite{adam} is used with the obtained training and validation data sets and with a learning rate of $\alpha=10^{-3}$ and the exponential decay rates set to $\beta_1 = 0.9$ and $\beta_2 = 0.999$. The quality of the obtained models are assessed in terms of the \textit{normalized root mean square} (NRMS) and RMS errors:
% \begin{equation}
% \mathrm{NRMS}=\frac{\mathrm{RMS}}{\bar{\sigma}_{y}}=\frac{\text{mean}\left(\sqrt{\frac{1}{M-k_0+1} \sum_{k=k_{0}}^{M}\left\|\hat{y}_{k}-y_{k}\right\|_{2}^{2}}\right)}{\text{mean}(\sigma_{y})}
% \end{equation}
\begin{equation}\label{eq:NRMS}
\mathrm{NRMS}=\frac{\mathrm{RMS}}{\sigma_{y}}=\frac{\sqrt{\frac{1}{N-n} \sum_{k=n}^{N}\left\|\hat{y}_{k}-y_{k}\right\|_{2}^{2}}}{\sigma_{y}}
\end{equation}
where  $\sigma_y$ is the sample standard deviation of $y$, giving $\text{NRMS}\in[0,1]$. Note that the first $n$ steps are skipped in \eqref{eq:NRMS} as they are used for the encoder function. \par
%First, we analyse the simulation \tb{responses} of the resulting models on noiseless test data.
Fig.~\ref{fig:WH_simulation_results} shows the simulation performance of the identified models on noiseless test data when trained on noisy data of particular SNRs. Models with various choices of $K_\theta$ are also compared: no $K_\theta$, linear $K_\theta$, and general $K_\theta(\hat{z}_k,u_k,\hat{e}_k)$. It can be seen that the linear innovation noise structure is able to reduce the \revtwo{estimation} %approximation 
error %with a factor of up 
%to
\revtwo{by up to factor}
two and slight improvements are obtained with the more complex noise model. Next, we analyse the one-step-ahead prediction performance. Similar to the simulation test case, Fig.~\ref{fig:WH_filtering_results} compares the NRMS \revtwo{of the prediction performance} %errors 
of \revtwo{the estimated models w.r.t.} noisy test data. As in the simulation scenario, the models are trained on noisy data with the respective SNRs. To provide context for these results, two additional dashed reference lines are included in the figure which represent the one-step-ahead prediction error of the true model, showing the proven convergence and consistency properties as the noise diminishes. It is also clear that incorporating a linear $K_\theta$ significantly reduces the NRMS error compared to the baseline model (which is correctly identified) without an innovation noise structure. In overall, both the simulation and prediction results highlight the importance of the innovation noise structure in increasing accuracy of the identified Koopman models.\par

\subsection{Bouc-Wen benchmark} \label{sec:bouc_wen} 
The Bouc-Wen oscillator benchmark \cite{BW_dataset,NL_benchmark,Schoukens_3_bench} is extensively used for testing capabilities of nonlinear system identification approaches as it describes a system with hysteresis which is a challenging behavior to capture from data accurately. The Bouc-Wen oscillator can be modeled as:
\begin{equation}\label{eq:BW_nl_model}
	m_L\ddot{y} + r(y,\dot{y})+q(y,\dot{y})=u,
\end{equation}
with $m_L$ the mass, $y$ its displacement, $\dot{y}$ the velocity, and $u$ the external force applied to the system. The restoring force $r(y,\dot{y})$ is linear while $q(y,\dot{y})$ is a dynamic nonlinear function describing the hysteresis curve. Both these effects are extensively described in \cite{Schoukens_3_bench} where the parameters are chosen for the Bouc-Wen benchmark such that the hysteretic behavior is significantly present.
\par To estimate a Koopman model of the Bouc-Wen benchmark system, the training dataset contains input-output data generated by using two sinusoidal signals with frequencies of 1 and 4 Hz, as well as four random phase multisine signals with frequency bands ranging from 1 to 150 Hz. For validation, a dataset is generated with an input containing a multisine signal also with an excited frequency range between 1 and 150 Hz followed by a sinusoidal signal at 2 Hz. For test data, the simulated response is generated by multisine and sinesweeep signals as given in \cite{Schoukens_3_bench}, as well as a 3 Hz sinusoidal signal to showcase the hysteretic behavior. Note that, as detailed in \cite{Schoukens_3_bench}, the data is sampled at a frequency of $f_\mathrm{s}=750$ Hz. These used detasets are shown in Fig. \ref{fig:BW_train_val_test}. As the data is noiseless, the estimation of a noise model is not required which means that $K_\theta$ is set to $0$. \par 
To estimate a Koopman model of the system, the encoder $\bar\Pi^\eta_{\mathrm{z},n}$ and the input matrix function $B_\theta$ are parametrized as feedforward neural networks with one hidden layer, having tanh activation and 40 neurons and we use $n=5$ number of lags. The network parameters are initialized as in the previous example. The prediction horizon value is set to $T=101$ with a batch size 256. For training, a learning rate of $\alpha=10^{-3}$ and exponential decays $\beta_1=0.9$ and $\beta_2=0.999$ are used. \par
\begin{figure}[tb!]
    \centering
    \begin{minipage}[b]{0.87\linewidth}
        \centering
        \includegraphics[width=\linewidth]{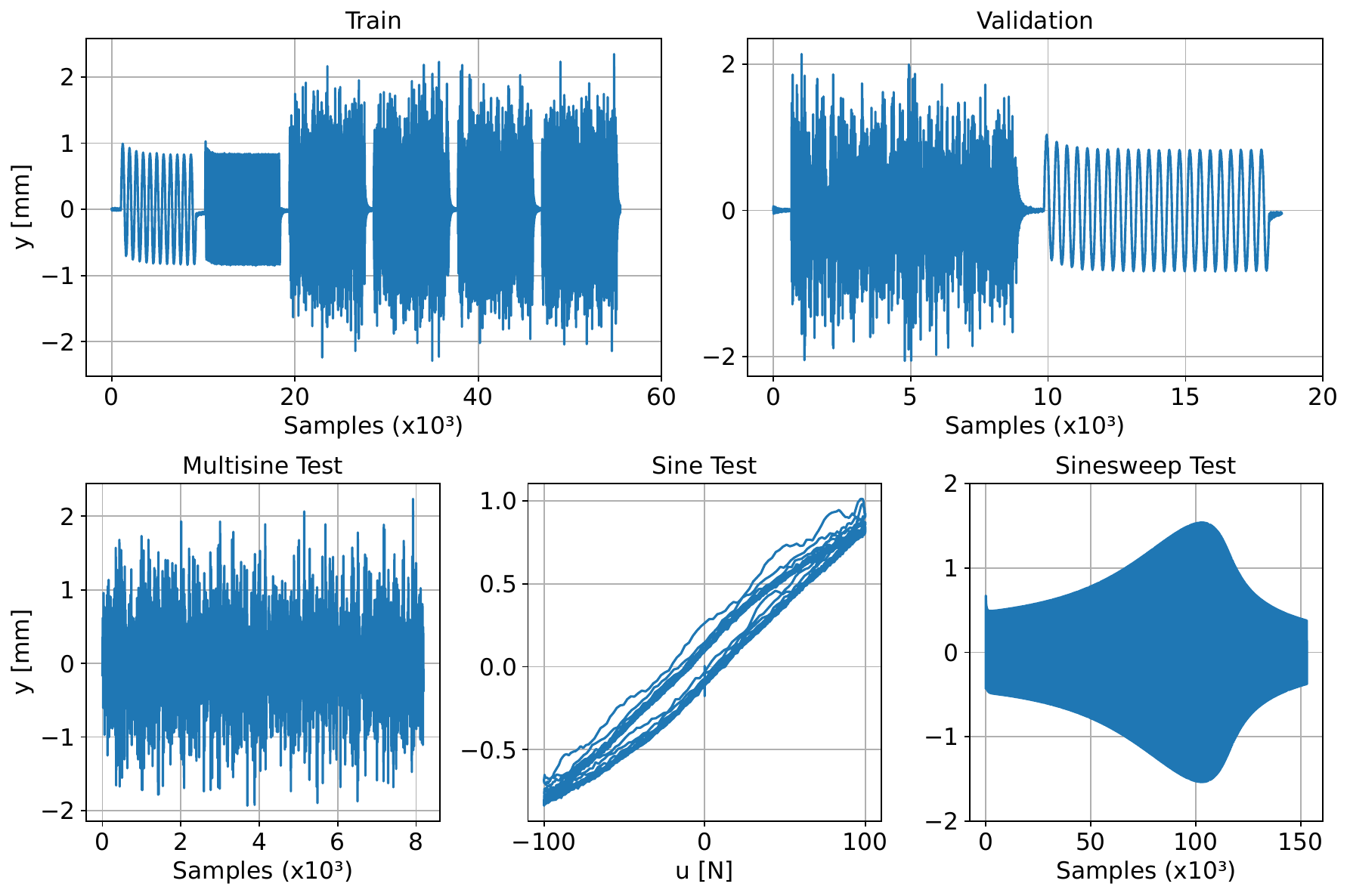}
    \end{minipage}
    \caption{Train, validation and test (multisine, sine and sinesweep) datasets used for the experiments (Bouc-Wen benchmark).}
    \label{fig:BW_train_val_test}
\end{figure}

\begin{figure}[tb!]
    \centering
    \begin{minipage}[b]{0.87\linewidth}
        \centering
        \includegraphics[width=\linewidth]{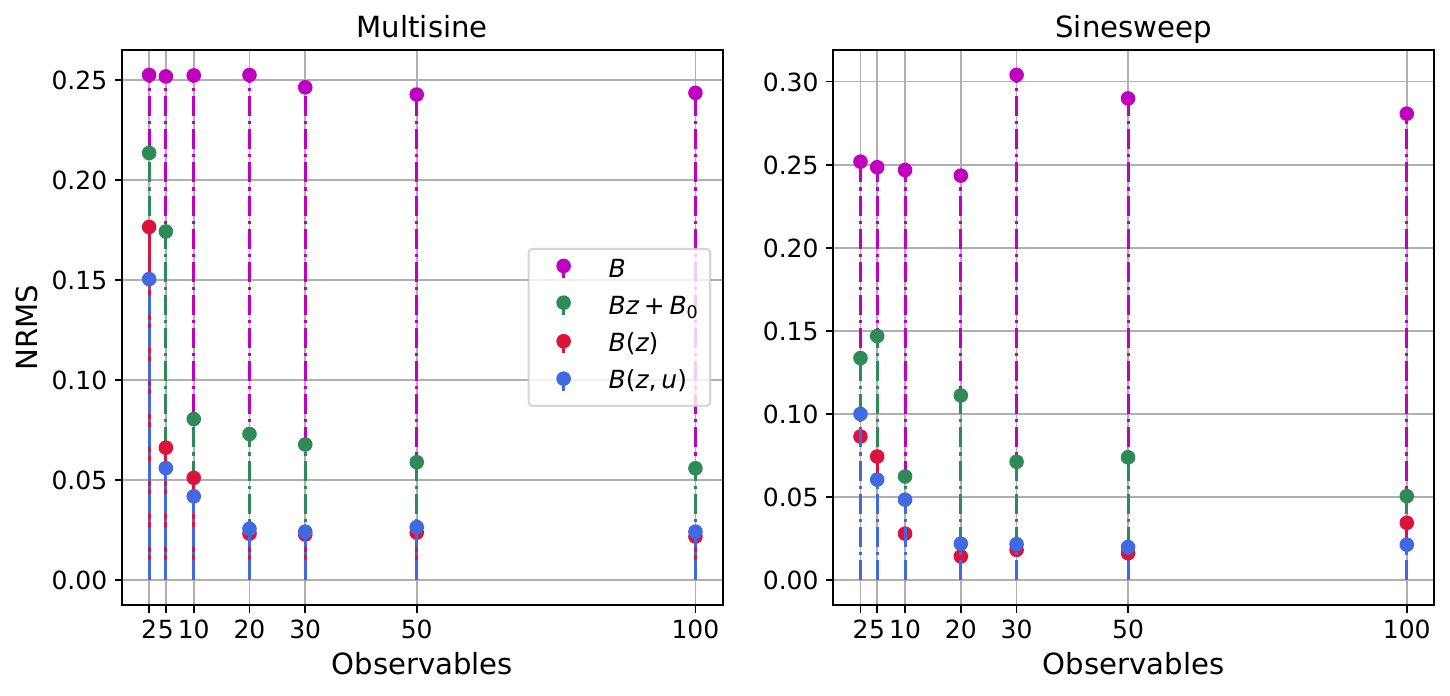}
    \end{minipage}
    \caption{Overview of NRMS errors of identified Koopman models with different complexities of $B$ (linear, bilinear, input affine, and general) and increasing lifting dimension using the multisine and sinesweep test datasets (Bouch-Wen benchmark).}
    \label{fig:BW_results}
\end{figure}

\begin{figure}[tb!]
    \centering
    \begin{minipage}[b]{\linewidth}
        \centering
        \includegraphics[width=\linewidth]{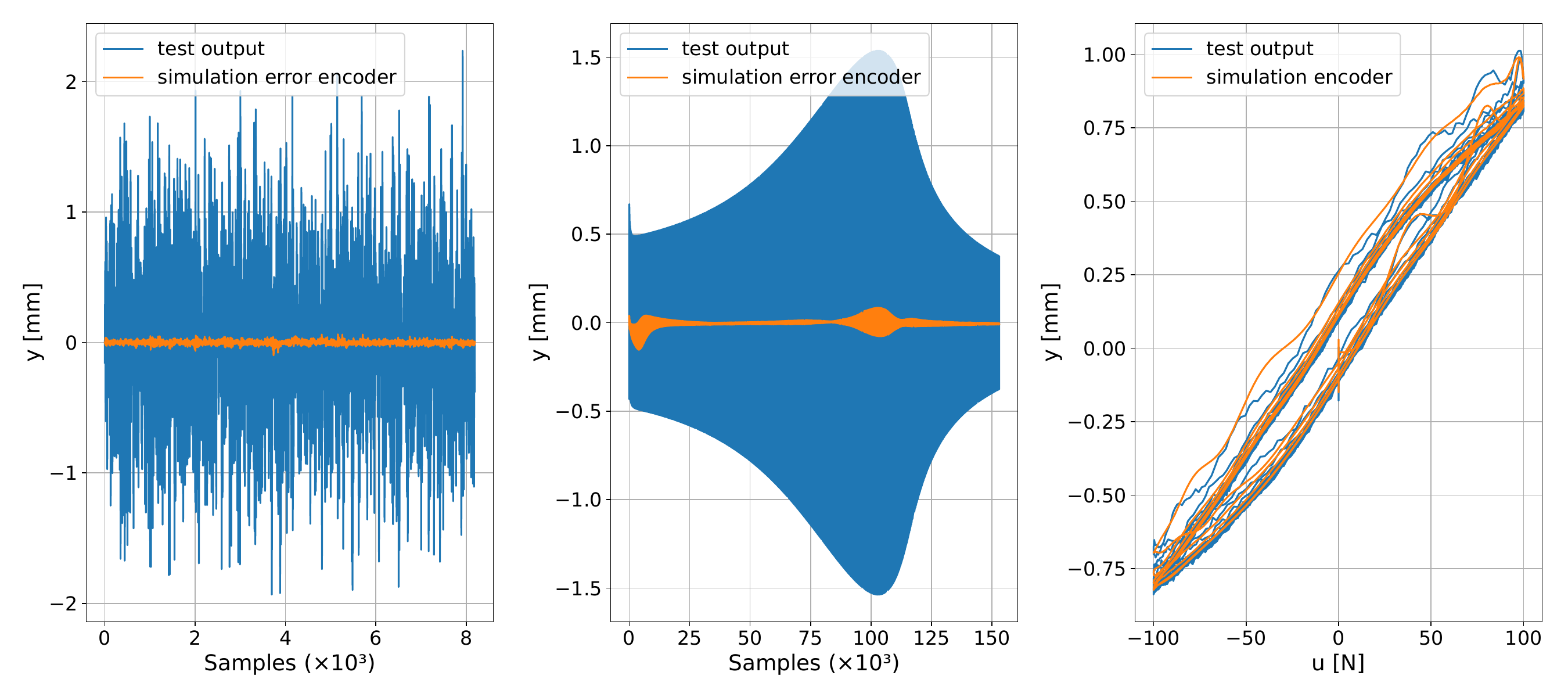}
    \end{minipage}
    \caption{Simulated output responses of the estimated model ($B(z)$ structure with $n_\mathrm{z}=100$) on the test data under multisine (left) and sinesweep (centre) inputs and hysteretic behavior (right) (Bouch-Wen benchmark).}\vspace{-.5cm}
    \label{fig:BW_subplots_results}
\end{figure}

Next, we show how different complexities in the $B$ function as well as the lifting dimension affect the approximation capabilities of Koopman models. This is shown in Fig.~\ref{fig:BW_results}, for both the multisine and sinesweep tests, using a number of observables $n_\mathrm{z}\in\{3,5,10,20,30,50,100\}$. The linear Koopman model with constant $B$ performs the worst, showing no significant improvement for larger $n_\mathrm{z}$. The \emph{bilinear} (BLTI) model shows a strong increase in accuracy with larger lifting dimensions however, the overall improvement from $n_\mathrm{z}=10$ to $n_\mathrm{z}=100$ drastically slows down. The best performing models are the input affine (with $B(z)$) and  full dependency (with $B(z,u)$) models, which demonstrate good approximation capability at only a relatively small lifting dimension, (e.g., $n_\mathrm{z}=20$). It can be seen that one can trade complexity with lifting dimension and vice-versa. For example, a bilinear model is generally a good trade-off between lifting dimension and approximation capability, while better approximation results can be obtained with input affine or general models at lower dimensions (e.g. $n_\mathrm{z}=20$ in this example) at the cost of model complexity. We do note that, somewhat nonintuitively, we obtain the lowest error with the input affine model (see Table \ref{table:BW_comparison}), instead of the general model structure, which is due to the increased size of the parameerization and complexity of the optimization problem. Moreover, as we utilize early stopping, it is possible that running the learning \revtwo{process} longer or optimizing the learning rates may produce \revtwo{slightly} better results.  However, a general conclusion is that a fully LTI \revtwo{Koopman} model is unable to capture the system dynamics. \par
The best performing model in the multisine case is the model structure with input affine complexity and a lifting dimension of $n_\mathrm{z}=100$ (although a close result is also obtained with $n_\mathrm{z}=20$). We use this model for comparing the results against other available methods. In Fig.~\ref{fig:BW_subplots_results}, the simulation results using the test data show a low error for the multisine and sinesweep input excitation, and clearly show that the hysteretic behavior is well captured. In Table \ref{table:BW_comparison}, we can see that the simulation performance of the obtained Koopman model is close to the state of the art (for the interested reader, the other methods are described in more detail in \cite{PLNLSS,Improved_NN_Maarten}). The approaches that obtain slightly better results do not impose a particular structure on the learned model. The Koopman model \eqref{eq:koop_model_structure} is able to accurately capture the system behavior and offers a good overall trade-off between state dimension and model complexity. Furthermore, even the more complex structures, i.e., input affine or general, can be cast into LPV representations \cite{Aut_Iacob_inputs} for which there exist convex and computationally efficient control methods \cite{Mohammadpour}. 

\begin{table}[tbp]
\resizebox{\textwidth}{!}{\begin{tabular}{||c||c|c|c|c||}
\hline Method & gr-SS-NN & SS-NN Suykens Impr & SS-NN Suykens & Poly-NL-SS \\
\hline \hline RMS  & $7.53 \times 10^{-6}$ & $8.91 \times 10^{-6}$ & $2.65 \times 10^{-5}$ & $1.34 \times 10^{-5}$ \\
\hline \hline \hline Koopman structure & Linear & Bilinear & Input affine & General \\ 
\hline\hline RMS & $1.60 \times 10^{-4}$ & $3.69 \times 10^{-5}$ & $1.43 \times 10^{-5}$ & $1.59 \times 10^{-5}$\\ 
\hline\hline Nr. observables & $n_\mathrm{z}=50$ & $n_\mathrm{z}=100$ & $n_\mathrm{z}=100$ & $n_\mathrm{z}=50$\\ 
\hline
\end{tabular}}
\caption{Comparison  of the estimated Koppman models with state of the art \revtwo{identification} methods applied on the Bouch-Wen benchmark in terms of achieved simulation RMS on multisine test data.}\vspace{-.5cm}
\label{table:BW_comparison} \vspace{-4mm}
\end{table}

%%%%%%%%%%%%%%%%%%%%%%%%%%%%%%%%%%%%%%%%%%%%%%%%%%%%%%%%%%%%%%%%%%%%%%%%%%%%%%%%  

\subsection{Quadrotor example}\label{sec:quadrotor_exp}

In this section we demonstrate the applicability of the proposed Koopman model identification approach on capturing the flight dynamics of a Crazyflie 2.1 quadrotor. We first show simulation results, followed by an experimental study on the real-world system.

\subsubsection{Simulation study}\label{sec:quad_sim_exp}

The considered simulation model of the drone implements the rigid body dynamics as described in \cite{mahony_multirotor} and uses three coordinate frames: \emph{north-east-down} (NED) oriented inertial frame $F_\mathrm{i}$; the vehicle frame $F_\mathrm{v}$ (origin  at the centre of gravity of the quadrotor) sharing the same orientation as $F_\mathrm{i}$ ; the body frame  $F_\mathrm{b}$ (orientation fixed to the quadrotor) whose origin coincides with $F_\mathrm{v}$. The model has 12 motion states composed of the position $\revtwo{s} = [\ \revtwo{s_\mathrm{x}} \ \ \revtwo{s_\mathrm{y}} \ \ \revtwo{s_\mathrm{z}}\ ]^\top$, translational velocity $v = [\ v_\mathrm{x} \ \  v_\mathrm{y} \ \  v_\mathrm{z}\ ]^\top$ expressed in $F_\mathrm{i}$,  $\zeta = [ \ \phi \ \  \theta \ \ \psi\ ]^\top$ describing the orientation as Z-Y-X Euler angles in $F_\mathrm{v}$, and $\omega = [ \ p \ \ q \ \ r ]^\top$, describing the rotational velocity of $F_\mathrm{b}$ w.r.t.~$F_\mathrm{v}$, given in $F_\mathrm{b}$. The inputs to the system are the total thrust $T$ and the torque vector $\tau = [\ \tau_\phi \ \ \tau_\theta \ \ \tau_\psi]^\top$ both given in $F_\mathrm{b}$ and produced by the four rotors. The diagonal values of the inertia matrix of the drone are set to \( J_\mathrm{x} = J_\mathrm{y} = 1.4 \times 10^{-5}\, \text{kg}\!\cdot\!\text{m}^2 \), \( J_\mathrm{z} = 2.17 \times 10^{-5}\, \text{kg}\!\cdot\!\text{m}^2 \),  and the off-diagonal values are 0, the mass is \( m = 0.027\, \text{kg} \), and \( g = 9.8\, \text{m/s}^2 \), which are experimentally obtained using a real-world Crazyflie 2.1 quadrotor. The simulation is performed with Runge-Kutta 4 integration at a sampling rate of 48 Hz, accurately replicating the expected flight-dynamics.

As the system is unstable, flight-trajectories are generated by using a gain-scheduled \emph{linear quadratic regulator} (LQR) controller, designed w.r.t. the local linearisations at each time step of the simulation model. 
The LQR is scheduled based on a desired state trajectory that is calculated for a $\revtwo{s_\mathrm{x}}$-$\revtwo{s_\mathrm{y}}$-$\revtwo{s_\mathrm{z}}$-$\psi$ defined flight-path reference by taking advantage of the differential flatness property of the system, see \cite{mellinger_control}. 
We presume full state measurements, hence the collected dataset consists of the system inputs and the states of the dynamical system. Since the dynamics governing the evolution of the position states consist solely of integration, the identification process can be simplified by focusing only on \([ \ v^\top \ \ \zeta^\top \ \ \omega^\top\ ]^\top\). The size of the recorded datasets can be viewed in Table \ref{tab:datasets}.
%the velocity dynamics, which describe the changes in the translational and rotational velocities, as well as the orientation states. 
For more details about the simulation and data collection procedure, the reader can refer to \cite{deep_learning_of_vehicle_dyn}. \par
In our study, we have investigated various dependencies of the input matrix \( B \) and emulated different levels of sensor noise. Additionally, we examined the effect of including the original motion states among the observables by setting \(C = [\ I \ \ 0 \ ]\) where \( I \) and \( 0 \) denote identity and zero matrices of appropriate dimensions. This modification aids the design of reference tracking controllers, as it allows the reference signal to be defined directly in terms of the original motion states.

To train the Koopman models, a prediction horizon of \(T=80\) is selected for $V^{\mathrm{batch}}_{\mathcal{D}_N}$, corresponding to 1.7 seconds of flight. This duration is sufficient as it significantly exceeds the largest time constants of a miniature quadrotor. During experiments, we found that a lifted state dimension of \(n_\mathrm{z}=40\) works best. Dimensions lower than this failed to adequately capture the system behavior, whereas dimensions higher than this led to overfitting, as evidenced by elevated NRMS values during model testing. The encoder $\bar\Pi^\eta_{\mathrm{z},0}$ and the input matrix function \(B_{\theta}\) are parametrized as deep neural networks with 2 hidden layers, 64 nodes, and tanh activation. With the availability of full-state measurements, the encoder only uses one measurement corresponding to the present timestep, and so it is simplified to be the lifting function. Parameter initialization is done identically to the previous experiments. For optimization, a batch size of 256 is selected, with a learning rate of \(\alpha=10^{-4}\) and exponential decay rates of \(\beta_1 = 0.9\) and \(\beta_2 = 0.999\). Additionally, an \(\ell_2\) parameter regularization is added to \eqref{eq:pred:enc} to prevent overfitting and we set the penalty coefficient to \(\lambda=0.5 \times 10^{-4}\).\par
Model performance across various \(n_\mathrm{z}\) values can be seen in Table \ref{tab:quad_nrms_errors} in terms of the 160-step NRMS errors for the trained models on the test data set. For comparison purposes, results on a full nonlinear model estimated by the SUBNET approach \cite{deep_learning_of_vehicle_dyn} are also included to indicate performance of a nonlinear model estimate without any restrictions on the network structure. In the network structure column, the subscripts \((\cdot)_\theta\), \((\cdot)_{\text{lin}}\), and \((\cdot)_{\mathbf{I}}\) denote the implementation of a function as a deep neural network, linear or identity layers. The superscript of \(A^{(\cdot)}\) denotes the lifted state-space dimension. The Koopman model achieves good performance in terms of the NRMS error %\% ($\text{NRMS\%}=\text{NRMS}\times 100$) errors 
only slightly exceeding that of the nonlinear SUBNET. Similar to the Bouc-Wen example, we found that the dependence of \(B\) on \(u\) slightly deteriorates the performance, compared to the input affine model structure. In the table, the effects of measurement noise can also be seen. Even at a level of 20 dB, the Koopman model remains capable of capturing the dynamics. Enforcing the original states among the observables only slightly decreases the simulation precision of the model, which we consider as a good tradeoff for the simpler model structure.

\renewcommand{\arraystretch}{1.1}
\begin{table}[tbp]
    \centering
\begin{tabular}{ | c | c | c || c | c | c | }
        \hline
        \multirow{2}{*}{Data} & \multirow{2}{*}{Network structure} & \multirow{2}{*}{Noise} & \multicolumn{1}{c|}{RMS} & \multicolumn{2}{c|}{NRMS} \\
        \cline{4-6}
        & & & Train & Validation & Test \\
        \hline
        \hline
        \multirow{8}{*}{\rotatebox{90}{Simulation}} & $A^{20}$, $B_\theta(z)$, $C_{\text{lin}}$ & None & 0.108 & 0.120 & 0.130 \\
        \cline{2-6}
        & $A^{40}$, $B_\theta(z)$, $C_{\text{lin}}$ & None & 0.061 & 0.069 & 0.079 \\
        \cline{2-6}
        & $A^{60}$, $B_\theta(z)$, $C_{\text{lin}}$ & None & 0.071 & 0.096 & 0.089 \\
        \cline{2-6}
        & $A^{40}$, $B_\theta(z)$, $C_{\text{lin}}$ & 25 dB & 0.106 & 0.120 & 0.089 \\
        \cline{2-6}
        & $A^{40}$, $B_\theta(z)$, $C_{\text{lin}}$ & 20 dB & 0.116 & 0.161 & 0.099 \\
        \cline{2-6}
        & $A^{40}$, $B_\theta(z, u)$, $C_{\text{lin}}$ & None & 0.074 & 0.090 & 0.087 \\
        \cline{2-6}
        & $A^{40}$, $B_\theta(z)$, $C_{\mathbf{I}}$ & None & 0.069 & 0.095 & 0.087 \\
        \cline{2-6}
        & General nonlinear & None & 0.061 & 0.042 & 0.055 \\
        \hline
        \hline
        \multirow{4}{*}{\rotatebox{90}{Real}} & $A^{40}$, $B_\theta(z)$, $C_{\text{lin}}$ & Sensor\footnotemark[5] & 0.163 & 0.213 & 0.228 \\
        \cline{2-6}
        & $A^{40}$, $B_\theta(z, u)$, $C_{\text{lin}}$ & Sensor & 0.151 & 0.199 & 0.217 \\
        \cline{2-6}
        & $A^{40}$, $B_\theta(z)$, $C_{\mathbf{I}}$ & Sensor & 0.140 & 0.207 & 0.216 \\
        \cline{2-6}
        & General nonlinear & Sensor & 0.142 & 0.200 & 0.216 \\
        \hline
    \end{tabular}
    \caption{Precision of various identified Koopman models in the quadrotor example. The used model structures and the level of added noise to the datasets are specified along with the required training time and NRMS errors w.r.t.~both simulation and experimental datasets.}\vspace{-.5cm}
    \label{tab:quad_nrms_errors}
    \vspace{-10pt}
\end{table}
% \footnotetext[3]{\tb{The training was performed on a desktop PC equipped with an Intel(R) Xeon(R) W-2123 CPU and 32 GB RAM.}}

\subsubsection{Experimental results}\label{sec:quad_real_exp}
The experimental setup consists of the Crazyflie 2.1 nano quadrotor, equipped with onboard sensors and a microcontroller, the OptiTrack motion capture system for accurate global position measurements, and a ground control PC. State estimation is done by an extended Kalman filter also performing sensor fusion of the various measurements.

Data collection is done at a control and sampling frequency of 200 Hz, which is higher than in the simulation environment, but necessary for the agile maneuvering of the real quadrotor. The size of the datasets are reported in Table \ref{tab:datasets}. For data collection and to conduct the experiments, we use the same reference paths as in the simulation environment, executed by the Mellinger controller \cite{mellinger_control} for reference tracking.
\footnotetext[5]{No artificial noise was added, the real dataset was only affected by the noise inherent to the inaccuracies of the various sensors equipped on the Crazyflie.}

\begin{figure}[htbp]
    \centering
    \begin{minipage}[b]{0.48\textwidth}
        \centering
        \captionsetup{type=table} 
        \begin{tabular}{ |@{\hspace{4pt}}c@{\hspace{4pt}}| |@{\hspace{4pt}}c@{\hspace{4pt}}| @{\hspace{4pt}}c@{\hspace{4pt}}|}
            \hline
            Dataset sizes & Synthetic & Real-world \\
            \hline
            Train & $123\,197$ & $228\,323$\\
            \hline
            Validation & $21\,740$ & $40\,292$\\
            \hline
            Test & $480$ & $3\,000$\\
            \hline
        \end{tabular}
        \caption{The sample size of train, validation, and test datasets (quadrotor example).}
        \label{tab:datasets}
    \end{minipage}%
    \hfill
    \begin{minipage}[b]{0.48\textwidth}
        \centering
        \includegraphics[width=\textwidth]{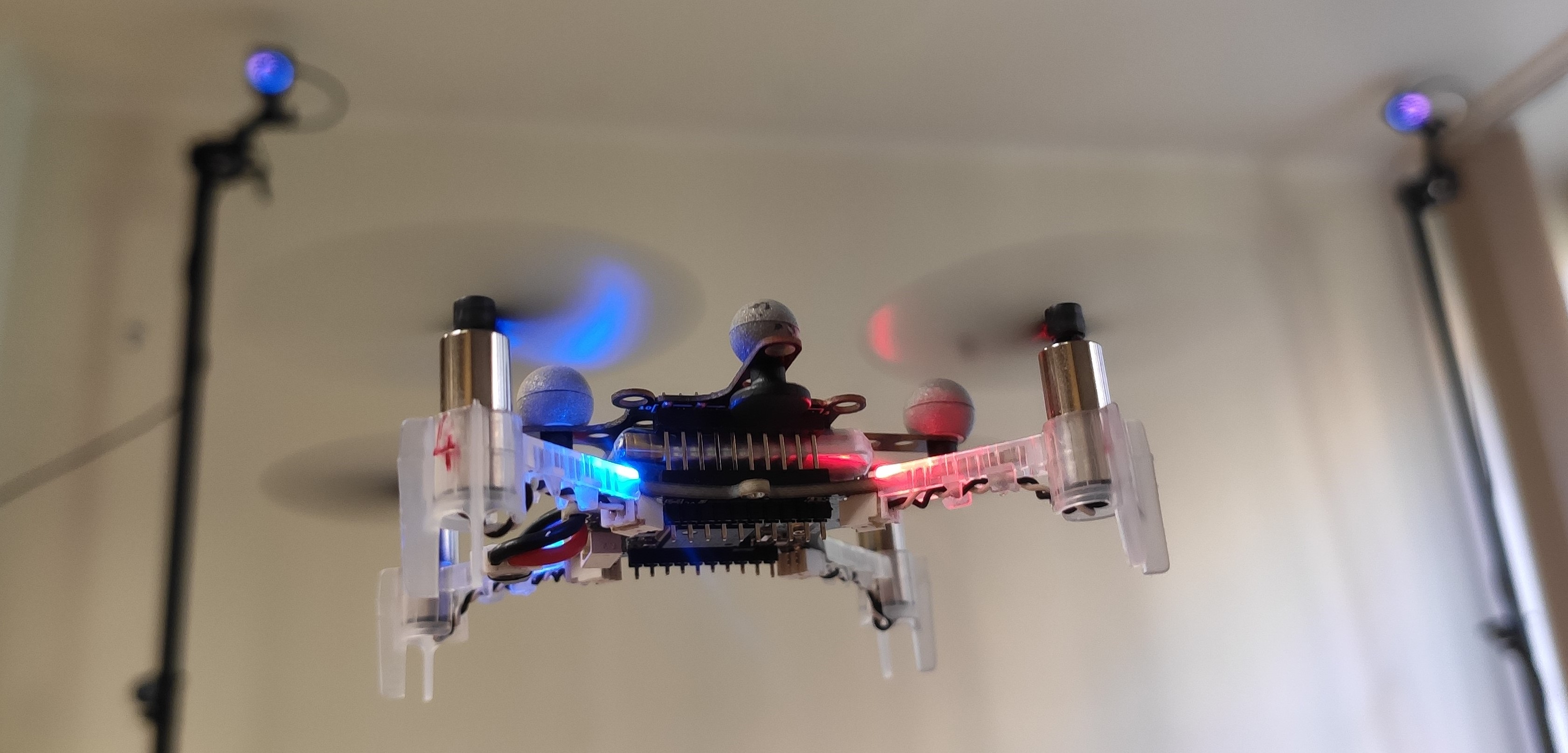}
        \caption{Crazyflie 2.1 during flight.}
        \label{fig:crazyflie}
    \end{minipage}
\end{figure}

The network structures that were found to perform the best in the simulation environment were used for training on the real dataset. The 200-step open loop test results can be viewed in Table \ref{tab:quad_nrms_errors}. The increase in NRMS values compared to simulation experiments may not fully represent the actual performance of the trained models. In open-loop simulations, deviations inevitably grow as errors accumulate over time and more critical is the resemblance between the shapes of the Koopman model outputs and the real system. Also, there are unmeasured deterministic disturbances such as airflow effects on the real quadcopter, which can not be fully captured by the considered noise model. 
As can be seen in Fig.~\ref{fig:Koopman_SUBNET_real_data} and Table \ref{tab:quad_nrms_errors}, the Koopman models closely match the performance of the general nonlinear SUBNET — in the identity output case they are identical — demonstrating that the Koopman models accurately capture the dynamics of the system.

\begin{figure}
    \centering
    \includegraphics[width=0.9\linewidth]{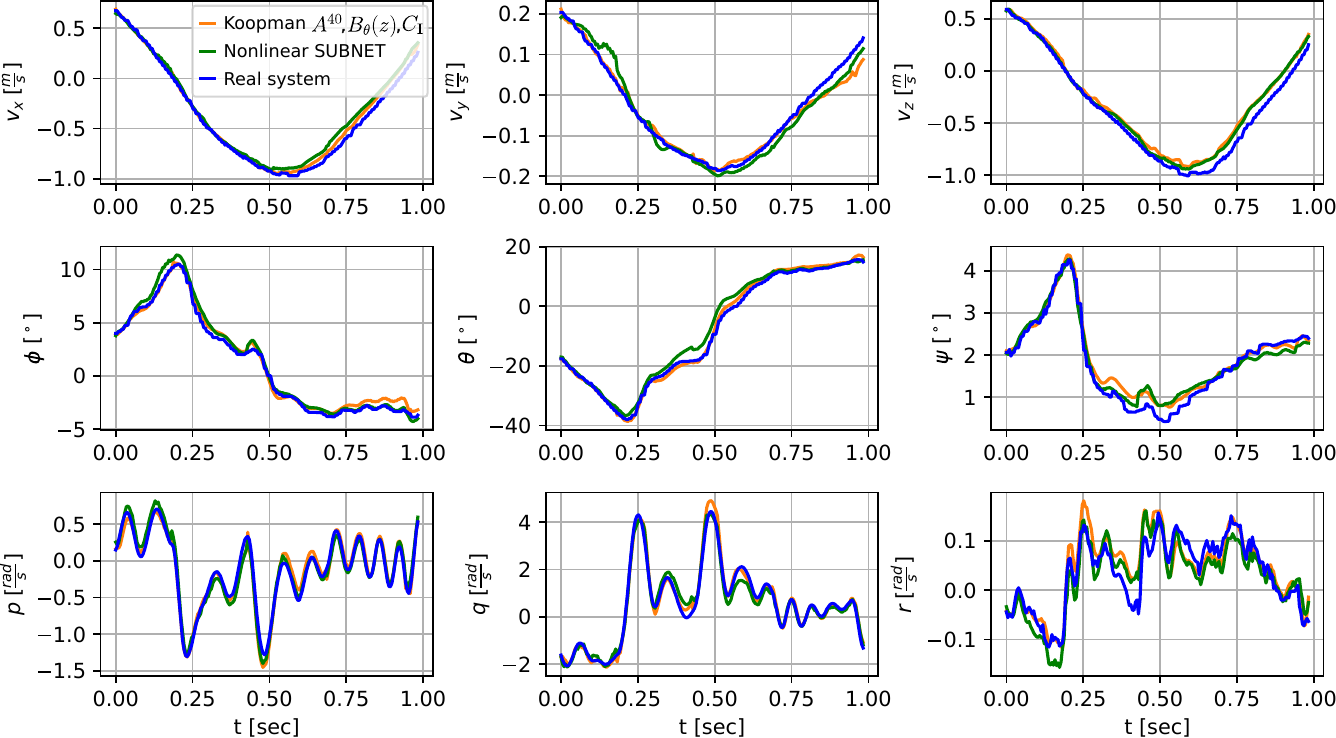}
    \caption{Simulation results of the estimated Koopman model and the estimated nonlinear SUBNET model w.r.t.~the measured flight data in the test data set.}\vspace{-.5cm}
    \label{fig:Koopman_SUBNET_real_data}
    \vspace{-4mm}
\end{figure}
%%%%%%%%%%%%%%%%%%%%%%%%%%%%%%%%%%%%%%%%%%%%%%%%%%%%%%%%%%%%%%%%%%%%%%%%%%%%%%%%

\section{Conclusion}\label{sec:section_conclusion}
In this paper a deep-learning-based Koopman identification method for nonlinear systems driven by an external input and affected by process and measurement noise is proposed, which provides statistically consistent estimation of the underlying nonlinear dynamics. For this purpose, we have shown that under control inputs and innovation noise, the data-generating system can be written in a Koopman model form \revtwo{both in the finite and infinite dimensional cases}, which in turn can be used for formulating a one-step-ahead predictor. With the help of this predictor and under various levels of complexity in the parameterization of the input and innovation matrices,  it has been shown that we can formulate a computationally efficient multiple-shooting-based learning method that minimizes the mean squared prediction error of the model.  
To circumvent a priori heuristic choice of a dictionary of observables, a neural-network-based encoder is used for the lifting and state-estimation, which is consistent with the reconstructability map of the Koopman model. Compared to other learning-based Koopman identification methods, the proposed approach not only provides theoretical guarantees of consistency and a computationally efficient learning pipeline even in case when no direct state measurments are available, but it is also shown to successfully capture the underlying nonlinear  behavior in various examples, from identification benchmarks to real-world flight dynamics of a quadcopter. 
%%%%%%%%%%%%%%%%%%%%%%%%%%%%%%%%%%%%%%%%%%%%%%%%%%%%%%%%%%%%%%%%%%%%%%%%%%%%%%%%
\appendix
\section{Proof of Theorem \ref{TH1}}\label{sec:appendix_proof_embedding}
To show that \eqref{eq:koop_model_structure} is an  exact embedding of \eqref{eq:data_gen}, we employ function factorization through the second \emph{fundamental theorem of calculus} (FTC), extending the approach in \cite{Aut_Iacob_inputs} to systems with process noise. Based on \eqref{eq:nl_sys_decomp}, we have the following decomposition: \vspace{-.2cm}
\begin{equation}
		x_{k+1}= f_\mathrm{d}(x_k,u_k,e_k)=f(x_k) + g(x_k,u_k)+d(x_k,u_k,e_k)\vspace{-.2cm}
\end{equation}
with $g(x_k,0)=0$ and $d(x_k,u_k,0)=0$. The proof is composed of three steps:

{\it Step 1: Embedding the autonomous part:}
Take $u_k=0$, $e_k=0$, implying $g(x_k,0)=0$ and $d(x_k,u_k,0)=0$ such that 
			$x_{k+1}=f(x_k)$.
		Then, based on Assumption \ref{assumption:exact_embedding_aut}, it holds that:\vspace{-.15cm}
		\begin{equation}\label{eq:appendix_aut_embedding}
			\Phi(x_{k+1})=\Phi(f(x_k))=A\Phi(x_k).\vspace{-.15cm}
		\end{equation}
        
{\it Step 2: Embedding the control input part:}  Take $e_k=0$, implying $d(x_k,u_k,0)=0$ such that
			$x_{k+1} = f(x_k) + g(x_k,u_k)$.
		Using the results in \cite{Aut_Iacob_inputs} and \eqref{eq:appendix_aut_embedding}, the exact lifted form of $x_{k+1} = f(x_k) + g(x_k,u_k)$ is: \vspace{.1cm}
		\begin{equation}\label{eq:appendix_lifting_no_noise}
			\Phi(x_{k+1})-\underbrace{\Phi(f(x_k))}_{A\Phi(x_k)}=\underbrace{\left(\int^1_0 \frac{\partial \Phi}{\partial x}(f(x_k)+\lambda \rev{g(x_k,u_k))}\dif \lambda\right)g(x_k,u_k)}_{\tilde B_\mathrm{x}(x_k,u_k)}.
		\end{equation}

        {\it Step 3: Embedding the noise part:}  
        For the full nonlinear dynamics described by:
	 	\begin{equation}\label{eq:expansion_nl_sys_with_noise}
			x_{k+1} = f(x_k) + g(x_k,u_k) + d(x_k,u_k,e_k),
		\end{equation}
        we can apply the proof of Theorem 2 in \cite{Aut_Iacob_inputs}: in Eq. (43) in \cite{Aut_Iacob_inputs}, choose $q_{k+1}=x_{k+1}$ (which is expanded as \eqref{eq:expansion_nl_sys_with_noise}) and $p_{k+1}=f(x_k)+g(x_k,u_k)$, giving an exact lifted form that includes the effect of noise as:
		\begin{multline}
			\Phi(x_{k+1})-\underbrace{\Phi(f(x_k)+g(x_k,u_k))}_{A\Phi(x_k)+\tilde{B}_\mathrm{x}(x_k,u_k)}=\\ =\underbrace{\left(\int^1_0\frac{\partial \Phi}{\partial x}(f(x_k)+g(x_k,u_k) + \lambda d(x_k,u_k,e_k))\dif \lambda\right)d(x_k,u_k,e_k)}_{\tilde{K}_\mathrm{x}(x_k,u_k,e_k)}.
		\end{multline}
By applying the exact factorization Lemma 1 in \cite{Aut_Iacob_inputs}, we get:
\begin{equation}
	\Phi(x_{k+1})=A\Phi(x_k) + B_\mathrm{x}(x_k,u_k)u_k + K_\mathrm{x}(x_k,u_k,e_k)e_k
\end{equation}
with
\begin{equation*}
	B_\mathrm{x}(x_k,u_k) = \int^1_0\frac{\partial \tilde{B}_\mathrm{x}}{\partial u}(x_k,\lambda u_k)\dif \lambda , \;\;K_\mathrm{x}(x_k,u_k,e_k)= \int^1_0\frac{\partial \tilde{K}_\mathrm{x}}{\partial e}(x_k, u_k, \lambda e_k)\dif \lambda.
\end{equation*}
Furthermore, let $z_k=\Phi(x_k)$ and $x_k=\Phi^\dagger(z_k)$, where $\dagger$ denotes the inverse. Then, considering that  Assumption \ref{assumption:exact_embedding_output_aut} holds, i.e., $h\in \text{span} \{\Phi\}$, an exact finite dimensional Koopman embedding of \eqref{eq:data_gen} is given by:
\begin{equation}
\begin{split}
	z_{k+1}&=Az_k + B(z_k,u_k)u_k + K(z_k,u_k,e_k)e_k\\
	y_k &= Cz_k + e_k
\end{split}
\end{equation}
with $B(z_k,u_k):=B_\mathrm{x}(\Phi^\dagger(z_k),u_k)$, $K(z_k,u_k,e_k):=K_\mathrm{x}(\Phi^\dagger(z_k),u_k,e_k)$.

    \revtwo{ \section{Proof of Lemma \ref{TH1}}\label{sec:appendix_inf_embedding}
    This result is actually an intermediary step to obtain \cref{TH1}. Instead of stacking observables into the dictionary $\Phi$ as detailed in \cref{sec:appendix_proof_embedding}, one can compute the dynamics per observable as described at length in \cite{Aut_Iacob_inputs} and apply the same 3-step procedure detailed in \cref{sec:appendix_proof_embedding}. The derivation leads to:
      \begin{equation}\label{eq:per_obs_dyna_calc_intermediary}
\begin{split}
			\phi_i(x_{k+1})=\phi_i(f(x_k))&+\left(\int^1_0\frac{\partial \phi_i}{\partial x}(f(x_k)+g(x_k,u_k))\dif \lambda\right)g(x_k,u_k)
            \\ &+\left(\int^1_0\frac{\partial \phi_i}{\partial x}(f(x_k)+g(x_k,u_k) + \lambda d(x_k,u_k,e_k))\dif \lambda\right)d(x_k,u_k,e_k)
            \end{split}
		\end{equation}
    and, using \eqref{eq:koop_composition}, the result directly follows. }

%%%%%%%%%%%%%%%%%%%%%%%%%%%%%%%%%%%%%%%%%
% \vspace{-1cm}
\bibliographystyle{siamplain}
\bibliography{siads_identification_references}

\end{document}